\documentclass[%
 reprint,
superscriptaddress,
 amsmath,amssymb,
 aps,
]{revtex4-2}

\usepackage{graphicx}
\usepackage{dcolumn}
\usepackage{bm}
\usepackage{braket}
\usepackage{mathrsfs}  
\DeclareMathOperator{\Tr}{Tr}
\usepackage{multirow}
\usepackage{relsize}
\usepackage{algorithm}
\usepackage{algpseudocode}
\usepackage[colorlinks=true, linkcolor=blue, citecolor=blue, urlcolor=black
]{hyperref}%
\usepackage{amsthm}
\usepackage{enumitem} 

\theoremstyle{definition}
\newtheorem{definition}{Definition}
\theoremstyle{plain}
\newtheorem{theorem}{Theorem}
\newtheorem{lemma}{Lemma}
\newtheorem{corollary}{Corollary}
\newcommand{\norm}[1]{\left\lVert#1\right\rVert}
\DeclareMathOperator*{\argmax}{arg\,max}
\DeclareMathOperator*{\argmin}{arg\,min}

\algrenewcommand\algorithmicrequire{\textbf{Input:}}
\algrenewcommand\algorithmicensure{\textbf{Output:}}

\begin{document}

\preprint{APS/123-QED}

\title{A kernel-based quantum random forest for improved classification}

\affiliation{School of Physics, University of Melbourne, VIC, Parkville, 3010, Australia.}
\affiliation{School of Mathematics and Statistics, University of Melbourne, VIC, Parkville, 3010, Australia.}


\author{Maiyuren Srikumar}
\email{srikumarm@unimelb.edu.au}
\affiliation{School of Physics, University of Melbourne, VIC, Parkville, 3010, Australia.}

\author{Charles D. Hill}
\email{cdhill@unimelb.edu.au}
\affiliation{School of Physics, University of Melbourne, VIC, Parkville, 3010, Australia.}
\affiliation{School of Mathematics and Statistics, University of Melbourne, VIC, Parkville, 3010, Australia.}
\author{Lloyd C.L. Hollenberg}
\email{lloydch@unimelb.edu.au}
\affiliation{School of Physics, University of Melbourne, VIC, Parkville, 3010, Australia.}

\date{\today}

\begin{abstract}

The emergence of Quantum Machine Learning (QML) to enhance traditional classical learning methods has seen various limitations to its realisation. There is therefore an imperative to develop quantum models with unique model hypotheses to attain expressional and computational advantage. In this work we extend the linear quantum support vector machine (QSVM) with kernel function computed through quantum kernel estimation (QKE), to form a decision tree classifier constructed from a decision directed acyclic graph of QSVM nodes - the ensemble of which we term the quantum random forest (QRF). To limit overfitting, we further extend the model to employ a low-rank Nystr\"{o}m approximation to the kernel matrix.
We provide generalisation error bounds on the model and theoretical guarantees to limit errors due to finite sampling on the Nystr\"{o}m-QKE strategy. In doing so, we show that we can achieve lower sampling complexity when compared to QKE. We numerically illustrate the effect of varying model hyperparameters and finally demonstrate that the QRF is able obtain superior performance over QSVMs, while also requiring fewer kernel estimations.


\end{abstract}

\maketitle


\section{\label{sec:level1}Introduction} 

The field of quantum machine learning (QML) is currently in its infancy and there exist, not only questions of advantages over its classical counter-part, but also questions of its realisability in real-world applications. Many early QML approaches \cite{Wiebe_2012, lloyd2013quantum} utilised the HHL algorithm \cite{Harrow_2009} to obtain speed-ups on the linear algebraic computations behind many learning methods. 
However, such methods presume efficient preparation of quantum states or \textit{quantum access} to data, limiting practicality \cite{10.1038/nphys3272}. 

Currently, the two leading contenders for near-term supervised QML methods are quantum neural networks (QNNs) \cite{10.1103/physreva.98.032309, farhi2018classification} and quantum kernel methods \cite{10.1038/s41586-019-0980-2, 10.1103/physrevlett.122.040504}.
QNNs constructed from parameterised quantum circuits embed data into a \textit{quantum feature space} and train parameters to minimise a loss function of some observable of the circuit. However, such variational algorithms face problems of barren plateaus in optimisation \cite{10.1038/s41467-018-07090-4, 10.1038/s41467-021-21728-w, PRXQuantum.3.010313} which hinder the feasibility of training the model. In contrast, quantum kernel methods suggest a non-parametric approach whereby only the inner product of quantum embedded data points are estimated on a quantum device -- a process known as \textit{quantum kernel estimation} (QKE) \cite{10.1038/s41586-019-0980-2, 10.1103/physrevlett.122.040504}. We refer to the model utilising QKE for constructing a \textit{support vector machine} (SVM), as a \textit{quantum}-SVM (QSVM) -- the most commonly explored quantum kernel method. 
Despite rigorous proofs of quantum advantage using a QSVM \cite{10.1038/s41567-021-01287-z}, they are not without their limitations for arbitrary \textit{real-world} datasets. Vanishing kernel elements \cite{10.1038/s41534-021-00498-9} subsequently require a large number of quantum circuit samples (suggesting a deep connection with barren plateaus \cite{inductive_bias_of_qk}) and indications that quantum kernels fail to generalise with the addition of qubits without the careful encoding of the correct problem-dependent \textit{inductive bias} \cite{inductive_bias_of_qk}.
Furthermore, the kernel approach quadratically grows in complexity with the number of training samples, as opposed to the linear scaling of QNNs. This becomes untenable for large datasets.
These limitations indicate that further work is required to practically deploy QML techniques. 

The close similarity of QNNs and quantum kernel methods \cite{qnns_are_kernel} highlight the importance of developing alternate QML models that present distinct model hypotheses.  
In this work, we propose a \textit{quantum random forest} (QRF) that fundamentally cannot be expressed through a kernel machine due to its discontinuous decision tree structure. Such a model is inspired by the classical \textit{random forest} (CRF) \cite{10.1023/A:1010933404324}.

The CRF has shown to be an extremely successful model in machine learning applications. However, there exist many barriers permitting the entry of an analogous QRF. The main challenge is the determination of an algorithm that is not limited by either the intrinsically quantum complication of the \textit{no-cloning theorem} \cite{10.1017/cbo9780511976667.010}, or the hard problem of efficiently loading data in superposition \cite{10.1038/nphys3272}. 
As opposed to previous works \cite{10.1007/s11128-013-0687-5, https://doi.org/10.48550/arxiv.2112.13346} that employ Grover's algorithm with input states of specific form, the QRF proposed in this work is ultimately an ensemble of weak classifiers referred to as \textit{quantum decision trees} (QDTs) with an approximate QSVM forming the \textit{split function}  -- shown in Figure \ref{fig:qrf_diagram}. Crucially, each QDT has random components that distinguishes it from the ensemble. The aim is to form a collection of uncorrelated weak QDT classifiers so that each classifier provides distinct additional information. The possible benefits of an ensemble of quantum classifiers have been expressed \cite{10.1038/s41598-018-20403-3}, with recent suggestions of boosting techniques for ensembles of QSVMs \cite{boosting_qsvm}. However, the distinctive tree structure has not yet been explored. Not only does this allow for discontinuous model hypotheses, we will see that the structure naturally accommodates multi-class supervised problems without further computation. In comparison, QSVMs would require \textit{one-against-all} (OAA) and \textit{one-against-one} (OAO) \cite{10.7551/mitpress/4175.001.0001} strategies which would respectively require constructing $|\mathcal{C}|$ and $|\mathcal{C}|(|\mathcal{C}|-1)/2$ separate QSVMs, where $\mathcal{C}$ is the set of classes in the problem.

The expressive advantage of the QRF is predicated on the distinguishability of data when mapped to a higher dimensional feature space. This proposition is supported by Cover's Theorem \cite{10.1109/pgec.1965.264137}, which states that data is more likely to be separable when projected to a sparse higher dimensional space. Cover's result is also a principal motivator for developing both quantum and classical kernel methods. However, it is often the case that generalisation suffers at the expense of a more expressive model -- a phenomena commonly known through the bias-variance trade-off in statistical learning. To counter the expressivity of the (already quite expressive) quantum model, we make a randomised low-rank Nystr\"{o}m approximation \cite{using_nystr_speed_kern} of the kernel matrix supplied to the SVM, limiting the ability of the model to overfit. We refer to this strategy as \textit{Nystr\"{o}m}-QKE (NQKE). We show that the inclusion of this approximation further allows for a reduced circuit sampling complexity.

The paper is structured as follows: in Section \ref{section:qrf_main} we address the exact form of the QRF model, before providing theoretical results in \ref{sec:theory}. This is followed by numerical results and discussion in Section \ref{section:results}.
Background to many of the ML techniques are presented in Appendix \ref{section:background}.

\section{The Quantum Random Forest model}\label{section:qrf_main}

\begin{figure*}
    \centering
    \includegraphics[width=500pt]{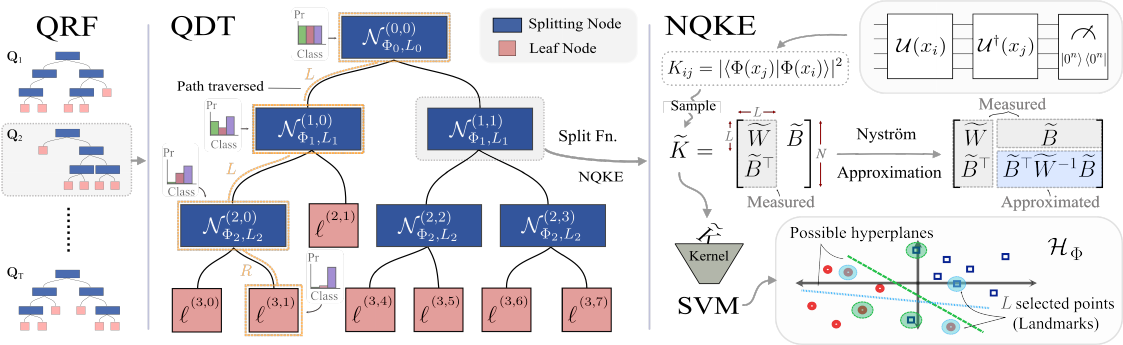}
    \caption{The quantum random forest (QRF) constitutes an ensemble of classifiers, referred to as quantum decision trees (QDTs). This tree structure is in turn a directed graph structure of nodes defined by the split function  $\mathcal{N}_{\Phi, L}:\mathbb{R}^D \xrightarrow{}\{-1,+1\}$ and $\Phi, L$ labels the type of function. The split function is an SVM that admits a Nystr\"{o}m approximated quantum kernel, $\widetilde{K}$. The kernel is defined through a chosen embedding $\Phi$ where $K_{ij}=|\langle \Phi(x_j) |\Phi(x_i)\rangle|^2$ and $\ket{\Phi(x_i)}=\mathcal{U}(x_i)\ket{0}$ for parameterised unitary $\mathcal{U}$. The randomly selected $L$ points for Nystr\"{o}m approximation determine the hyperplane generated in the feature space, $\mathcal{H}_\Phi$.}
    \label{fig:qrf_diagram}
\end{figure*}

The Quantum Random Forest (QRF) is composed of a set of $T$ distinct quantum decision trees (QDTs), $\{ \mathbf{Q}_t \}_{t=1}^{T}$, which form the weak independent classifiers of the QRF ensemble. This is illustrated in Figure \ref{fig:qrf_diagram}. As with all supervised ML methods, the QDTs are given the set of $N$ training instances, $x^{\mathrm{train}}_i \in \mathbb{R}^D$, and their associated $n_c$ class labels, $y^{\mathrm{train}}_i \in \mathcal{C} = \{0, 1, ..., n_c - 1 \}$, to learn from the annotated data set, $\mathcal{S} = \{ (x^{\mathrm{train}}_i, y^{\mathrm{train}}_i)\}_{i=1}^{N}$ sampled from some underlying distribution $\mathcal{D}$. 
Each tree is trained using $N_p \leq N$ -- which we will refer to as the partition size -- instances sampled from the original dataset $\mathcal{S}$. This acts as both a regularisation strategy, as well as a way to reduce the time complexity of the model. Such a technique is known as bagging \cite{10.5555/3086952} or bootstrap aggregation and it has the ability to greatly minimise the probability of overfitting. 

 
Once each classifier is trained, the model is tested with a separate set, $\mathcal{S}^{\mathrm{test}}$ of size $N'$, from which we compare the predicted class distribution with the real. To obtain the overall prediction of the QRF model, predictions of each QDT is averaged across the ensemble. Specifically, each QDT returns a probability distribution, $\mathbf{Q}_t(\vec{x}; c) = \mathrm{Pr}(c | \vec{x})$ for $c \in \mathcal{C}$, with the overall prediction as the class with the greatest mean across the ensemble. Hence, given an input, $\vec{x}$, we have prediction of the form,
\begin{equation}\label{eq:voting_qrf}
    \widetilde{y} (\vec{x}) = \argmax_{c \in \mathcal{C}} \Bigg\{  \frac{1}{T}\sum_{t=1}^T  \mathbf{Q}_t(\vec{x}; c)  \Bigg\} .
\end{equation}
\noindent The accuracy of the model can be obtained with the comparison of predictions with the real labels, $1/N'\sum_{(\vec{x}, y) \in \mathcal{S}^{\mathrm{test}}}  \delta_{y, \widetilde{y}(\vec{x})}$ where $\delta_{i, j}$ is the Kronecker delta. The distinction between training and testing data is generally critical to ensure that the model is not overfitting to the trained dataset and hence generalises to the prediction of previously unknown data points. We now elaborate on the structure of the QDT model.

\subsection{A Quantum Decision Tree}

The Quantum Decision Tree (QDT) proposed in this work is an independent  classification algorithm that has a directed graph structure identical to that of a binary decision tree -- illustrated in Figure \ref{fig:qrf_diagram} and elaborated further in Appendix \ref{section:rf_model}. Each vertex, also referred to as a \textit{node}, can either be a split or leaf node, distinguished by colour in Figure \ref{fig:qrf_diagram}. A leaf node determines the classification output of the QDT, while a split node separates inserted data points into two partitions that subsequently continue down the tree. The split effectiveness is measured by the reduction in entropy, referred to as the \textit{information gain} (IG). Given a labelled data set, $\mathcal{S}$, that is split into partitions $\mathcal{S}_L$ and $\mathcal{S}_R$ by the split function, we have,
\begin{equation}\label{eq:ig}
    \mathrm{IG}(\mathcal{S}; \mathcal{S}_L , \mathcal{S}_R ) = H(\mathcal{S}) - \sum_{i \in \{L, R\}} \frac{|\mathcal{S}_i|}{|\mathcal{S}|} H(\mathcal{S}_i),    
\end{equation}
\noindent where $H$ is the entropy over the class distribution defined as $H (\mathcal{S}) = -\sum_{c \in \mathcal{C}} \mathrm{Pr}(c) \log_2 \mathrm{Pr}(c)$ where class $c \in \mathcal{C}$ occurs with probability $\mathrm{Pr}(c)$ in the set $\mathcal{S}$. Clearly the IG will increase as the splitting more distinctly separates instances of different classes. Using information gain is advantageous especially when more than two classes are involved and an accuracy-based effectiveness measure fails.  

Given a partitioned training set at the root node, the QDT carries out this splitting process to the \textit{leaves} of the tree where a prediction of class distribution is made from the proportions of remaining data points in each class. Mathematically, with a subset of data points $\mathcal{S}^{(l)}$ supplied at training to leaf $\ell^{(l)}$ indexed by $l$, we set its prediction as the probability distribution,  
\begin{equation}\label{eq:leaf_prob}
  \ell^{(l)}(\mathcal{S}^{(l)}; c) = \frac{1}{\lvert \mathcal{S}^{(l)} \rvert}\sum_{(\vec{x}, y) \in \mathcal{S}^{(l)}} \big[ y = c  \big] ,
\end{equation}
\noindent where $[p ]$ is the Iversion bracket that returns 1 when the proposition $p$ is true and 0 otherwise. When training the model, a node is defined a \textit{leaf} if any of the following conditions are met, (i) the training instances supplied to the node are of a single class -- in which case further splitting is unnecessary, (ii) when the number of data points left for splitting is smaller than a user defined value, $m_s$, ie. $|\mathcal{S}^{(l)}| \leq m_s$, or (iii) the node is at the maximum depth of the tree, $d$, defined by the user. Once trained, prediction occurs by following an instance down the tree until it reaches a leaf node. The unique path from root to leaf is referred to as the \textit{evaluation path}. At this end point, the prediction of the QDT is the probability distribution held within the leaf, defined at training in equation \eqref{eq:leaf_prob}. The ensemble of predictions made by distinct QDTs are subsequently combined with Eq. \eqref{eq:voting_qrf}.

The crux of the QDT model lies at the performance of split nodes, with the $l^{\text{th}}$ node defined through a split function, $\mathcal{N}^{(l)}_{\theta}: \mathcal{X} \xrightarrow{} \{-1, +1\}$, where $\theta$ are (hyper-) parameters of the split function, that are either manually selected prior to, or algorithmically tuned during, the training phase of the model -- the specifics are elaborated in the next section. Instances are divided into partitions $\mathcal{S}_- = \{ (\vec{x}, y) \in \mathcal{S} | \mathcal{N}^{(l)}_{\theta}(\vec{x}) = -1 \}$ and $\mathcal{S}_+ = \{ (\vec{x}, y) \in \mathcal{S} | \mathcal{N}^{(l)}_{\theta}(\vec{x}) = 1 \}$. At this point it is natural to refer to the information gain of a split function as, $\mathrm{IG}(\mathcal{S}|\mathcal{N}^{(l)}_{\theta}):=\mathrm{IG}(\mathcal{S}; \mathcal{S}_- , \mathcal{S}_+ )$ with the aim being to maximise $\mathrm{IG}(\mathcal{S}|\mathcal{N}_{\theta})$ at each node.
In essence, we desire a tree that is able to separate points so that instances of different classes terminate at different leaf nodes. The intention being that a test instance that follows the unique evaluation path to a particular node is most likely to be similar to the instances that followed the same path during training. Hence, effective split functions amplify the separation of classes down the tree. However, it remains important that the split function generalises, as repeated splitting over a deep tree can easily lead to overfitting.
It is for this reason decision trees often employ simple split functions to avoid having a highly expressive model. Namely, the commonly used CART algorithm \cite{li1984classification} has a threshold function for a single attribute at each node. 

In the next section we specify the split function to have the form of a support vector machine (SVM) that employs a quantum kernel. The core motivation being to generate a separating hyperplane in a higher dimensional quantum feature space that can more readily distinguish instances of different classes. However, this is a complex method that will inevitably lead to overfitting.
Hence, we take an approximate approach to reduce its effectiveness without stifling possible quantum advantage.

\subsection{Nystr\"{o}m Quantum Kernel Estimation} 

The proposed split function generates a separating hyperplane through a kernelised SVM. The kernel matrix (or Gram matrix) is partially computed using Quantum Kernel Estimation (QKE) \cite{10.1103/physrevlett.113.130503} and completed using the Nystr\"{o}m approximation method. We therefore refer to this combined process as \textit{Nystr\"{o}m Quantum Kernel Estimation} (NQKE) with its parameters determining the core performance of a splitting node. Background on SVMs and a deeper discussion on the Nystr\"{o}m method are supplied in Appendix \ref{section:svm}. 

Given a dataset $\mathcal{S}^{(i)} = \{(x_j , y_j)\}_{j=1}^{N^{(i)}}$ to the $i^{\text{th}}$ split node, the Nystr\"{o}m method randomly selects a set of $L$ landmarks points $\mathcal{S}^{(i)}_L \subseteq \mathcal{S}^{(i)}$ . Without loss of generality we may assume $\mathcal{S}^{(i)}_L = \{(x_j , y_j)\}_{j=1}^{L}$. The inner product between elements in $\mathcal{S}^{(i)}$ and $\mathcal{S}^{(i)}_L$ are computed using a quantum kernel defined through the inner product of parameterised density matrices,
\begin{equation}
    k_{\Phi}(x', x'') = \Tr [\rho_{\Phi}(x') \rho_{\Phi}(x'')],
\end{equation}
where we have the spectral decomposition $\rho_{\Phi}(x) = \sum_j \lambda_j \ket{\Phi_j(x)} \bra{ \Phi_j(x)}$ with parameterised pure states $\{\ket{\Phi_j(x)}\}_j$ that determine the quantum feature map, $x\xrightarrow{} \rho(x)$. In practice, $\rho_{\Phi}$ are pure states reducing the trace to $k(x', x'')= |\langle\Phi(x'')| \Phi(x')\rangle|^2 = |\bra{0} \mathcal{U}(x'')^{\dagger}\mathcal{U}(x')\ket{0}|^2$, allowing one to obtain an estimate of the kernel by sampling the probability of measuring $\ket{0}$ on the $\mathcal{U}(x'')^{\dagger}\mathcal{U}(x')\ket{0}$ state. 

The Nystr\"{o}m method can be used for the matrix completion of positive semi-definite matrices. Using the subset of columns measured through the inner product of elements between $\mathcal{S}$ and $\mathcal{S}_L$ for an arbitrary node (dropping the the node label for brevity), which we define  as the $N\times L$ matrix $G := [W, B]^\top$ where $G_{ij} = k(x_i, x_j)$, $i\leq N$, $j\leq L$ and $W\in\mathbb{R}^{L\times L}$, $B\in\mathbb{R}^{L \times (N-L)}$, we complete the $N\times N$ kernel matrix by making the approximation $K \approx G W^{-1} G^\top$. Expanding, we have,
\begin{equation}
     K \approx \widehat{K} :=
    \begin{bmatrix}
    W      &  B \\
    B^\top      & B^\top W^{-1} B  
    \end{bmatrix},
\end{equation}
\noindent where, in general, $W^{-1}$ is the Moore-Penrose gneralised inverse of matrix $W$. This becomes important in cases where $W$ is singular. Intuitively, the Nystr\"{o}m method utilises correlations between sampled columns of the kernel matrix to form a low-rank approximation of the full matrix. This means that the approximation suffers when the underlying matrix $K$ is close to full-rank.
However, the \textit{manifold hypothesis} \cite{10.5555/3086952} suggests that choosing $L<<N$ is not entirely futile as data generally does not  explore all degrees of freedom and often lie on a sub-manifold, where a low-rank approximation is reasonable. 

The Nystr\"{o}m approximation does not change the positive semi-definiteness of the kernel matrix, $\widehat{K} \succeq 0$, and the SVM optimisation problem remains convex. Furthermore, the reproducing kernel Hilbert space (RKHS), induced by the adjusted kernel, $\hat{k}$, is a nonparametric class of functions of the form,
$
    \mathbb{H}_{\Phi} = \left\{ f_{\Phi} \big| f_{\Phi}(\cdot) = \sum_{i=1}^N \alpha_i'  \hat{k}_\Phi(\cdot, x_i) , \alpha_i' \in \mathbb{R}\right\}
$,
\noindent which holds as a result of the celebrated \textit{representer theorem} \cite{kimeldorf1971some, 10.1007/3-540-44581-1_27}. The specific set of $\{\alpha_i'\}_{i=1}^N$ for a given dataset is obtained through solving the SVM quadratic program. This constructs a split function that has the form,
\begin{equation}\label{eq:split_function}
    \mathcal{N}_{\Phi; \alpha} (x)=  \mathrm{sign}\left[ \sum_{i=1}^{N} \alpha_i \widetilde{y}_i \hat{k}_{\Phi}(x, x_i) \right],
\end{equation}
\noindent where $\alpha_i' = \alpha_i \widetilde{y}_i$, $\alpha_i \geq 0$ and $\widetilde{y}_i = F(y_i)$ with the function $F:\mathcal{C} \xrightarrow{} \{-1, 1\}$ mapping the original class labels, $y_i\in\mathcal{C}$ -- which in general can be from a set of many classes, i.e. $|\mathcal{C}|>2$, to a binary class problem. In Section \ref{section:results}, we provide numerical comparisons between two possible approaches of defining the function $F$, namely, (i) one-against-all (OAA), and (ii) even-split (ES) strategies -- see Appendix \ref{sec:multiclass_classification} for more details. 

The construction of a split function of this form adopts a \textit{randomised node optimisation} (RNO) \cite{10.1007/978-1-4471-4929-3} approach to ensuring that the correlation between QDTs are kept to a minimum. The randomness is injected through the selection of landmark data points, as the optimised hyperplane will differ depending on the subset chosen. Furthermore, there exist possibilities to vary hyperparameters, $\Phi$ and $L$ both, down and across trees. This is depicted in Figure \ref{fig:qrf_diagram} with split functions notated with a separate tuple $(\Phi_i, L_i)$ for depths $i=1,...,D-1$ of the tree.  
The specific kernel defined by an embedding $\Phi$ implies a unique Hilbert space of functions for which $k_{\Phi}$ is the reproducing kernel \cite{10.1090/s0002-9947-1950-0051437-7}. The two types of embeddings numerically explored are defined in Appendix \ref{section:embedding}. The QRF approach of employing distinct kernels at each depth of the tree, gives a more expressive enhancement to the QSVM method. A simple ensemble of independent QSVMs would require greater computational effort and would also lose the tree structure present in the QRF. 


\section{Theoretical Results}\label{sec:theory}

QML literature has seen a push towards understanding the performance of models through the well established perspective of statistical learning theory. There has recently been a great deal of work determining the generalisation error for various quantum models \cite{Caro_2020, 10.1038/s41467-021-22539-9,https://doi.org/10.48550/arxiv.2103.03139, PRXQuantum.2.040321} that illuminate optimal strategies for constructing quantum learning algorithms. This is often achieved by bounding the generalisation error (or risk) of a model producing some hypothesis $h \in \mathbb{H}$, $\mathcal{R}(h) = \Pr_{(x,y)\sim \mathcal{D}} [h(x) \neq y]$, where $\mathcal{D}$ indicates the underlying distribution of the dataset. However, $\mathcal{D}$ is unknown and one simply has a set of samples $\mathcal{S}\sim \mathcal{D}^N$. We therefore define the empirical error of a particular hypothesis $h$ as the error over known points, $\widehat{\mathcal{R}}(h) = \frac{1}{N}\sum_{(x_i , y_i) \in \mathcal{S}} \mathbf{1}_{h(x_i)\neq y_i}$. This is often referred to as the training error associated with a model and is equivalent to $(1-\text{accuracy})$. Achieving a small $\widehat{\mathcal{R}}(h)$ however does not guarantee a model capable of predicting new instances. Hence providing an upper bound to the generalisation error $\mathcal{R}(h)$ is imperative in understanding the performance of a learning model.

The QDT proposed in this work can be shown to generalise well in the case where the margins at split nodes are large. This is guaranteed by the following Lemma which builds on the generalisation error of \textit{perceptron decision trees} \cite{10.1023/a:1007600130808}, with proof supplied in Appendix \ref{sec:error_bounds}.
\begin{lemma}\label{lemma:gen_err_qdt} (Generalisation error of Quantum Decision Trees) 
Let $\mathbb{H}_J$ be the hypothesis set of all QDTs composed of $J$ split nodes. Suppose we have a QDT, $h\in \mathbb{H}_J$, with kernel matrix $K^{(i)}$ and labels $\{y^{(i)}_j\}_{j=1}^{N^{(i)}}$ given at the $i^{\text{th}}$ node. Given that $m$ instances are correctly classified, with high probability we can bound the generalisation error,
\begin{align}
 \mathcal{R}(h) \leq  \widetilde{\mathcal{O}}&\Bigg(\frac{1}{m} \Bigg[ J \log\left(4mJ^2\right) \nonumber \\ 
 &+\log(4m)^2 \sum_{i=1}^J \Big|\sum_{j,k=1}^{N^{(i)}} y^{(i)}_{j} y^{(i)}_{k}(K^{(i)+})_{jk} \Big| \Bigg]\Bigg),
\end{align}
where $K^{(i)}_{jk}=|\langle \Phi(x^{(i)}_j) | \Phi(x^{(i)}_k) \rangle|^2$ and $\{y_j^{(i)}\}_{j=1}^{N^{(i)}}$ are respectively the Gram matrix and binary labels given to the $i^{\text{th}}$ node in the tree, and $K^+$ denotes the Moore-Penrose generalised inverse of matrix $K$.
\end{lemma}
\noindent The term $s_K = \sum_{j,k=1}^{N^{(i)}} y^{(i)}_{j} y^{(i)}_{k}(K^{(i)+})_{jk}$ is referred to as the model complexity and appears in the generalisation error bounds for kernel methods \cite{10.1038/s41467-021-22539-9}. The term is both inversely related to the kernel target alignment \cite{Cristianini2006} measure that indicates how well a specific kernel function \textit{aligns} with a given dataset, as well as being an upper bound to the reciprocal of the geometric margin of the SVM. Therefore, Lemma \ref{lemma:gen_err_qdt} highlights the fact that, for a given training error, a QDT with split nodes exhibiting larger margins (and smaller model complexities) are more likely to generalise. Crucially, the generalisation bound is not dependent on the dimension of the feature space but rather the margins produced. This is in fact a common strategy for bounding kernel methods \cite{10.1007/s00362-019-01124-9}, as there are cases in which kernels represent inner products of vectors in infinite dimensional spaces \cite{10.7551/mitpress/4175.001.0001}. In the case of quantum kernels, bounds based on the Vapnik-Chervonenkis (VC) dimension would grow exponentially with the number of qubits. The result of Lemma \ref{lemma:gen_err_qdt} therefore suggests that large margins are analogous to working in a lower VC class. However, the disadvantage is that the bound is a function of the dataset and there are no \textit{a priori} guarantees that the problem exhibits large margins. 

The Nystr\"{o}m approximated quantum kernel matrix used in the optimisation of the SVM is also a low rank approximation, with estimated elements that are subject to finite sampling error. This prompts an analysis on the error introduced -- stated in the following Lemma.
\begin{lemma}\label{lemma:k_k_tilde_main}
Given a set of data points $\mathcal{S} = \{x_i\}_{i=1}^N$ and a quantum kernel function $k$ with an associated Gram matrix $K_{ij} = k(x_i, x_j)$ for $i,j=1,...,N$, we choose a random subset $\mathcal{S}_L \subseteq \mathcal{S}$ of size $L\leq N$ which we define as $\mathcal{S}_L = \{x_i\}_{i=1}^L$ without loss of generality. Estimating each matrix element with $M$ Bernoulli trials such that $\widetilde{K}_{ij} = (1/M)\sum_{p=1}^M \widetilde{K}_{ij}^{(p)} $ with $\widetilde{K}_{ij}^{(p)} \sim \mathrm{Bernoulli}(k(x_i, x_j))$ for $i\leq N$, $j\leq L$, the error in the Nystr\"{o}m completed $N\times N$ matrix, $\widetilde{K}$, can be bounded with high probability,
\begin{equation}\label{eq:k_k_tilde_bound}
    || K - \widetilde{K} ||_2 \leq \widetilde{\mathcal{O}} \Bigg(\frac{NL}{M} + \frac{N}{\sqrt{L}}\Bigg),
\end{equation}
\noindent where $\widetilde{\mathcal{O}}$ hides the $\log$ terms.
\end{lemma}
\noindent The proof is given in Appendix \ref{sec:error_bounds}. This Lemma indicates that there are two competing interests when it comes to reducing the error with respect to the expected kernel. Increasing the number of landmark points $L$ reduces the number of elements estimated and hence the error introduced by finite sampling noise. On the other hand, the Nystr\"{o}m approximation becomes less accurate as expressed through the second term in Eq. \eqref{eq:k_k_tilde_bound}.
However, it is important to realise that the approximation error of the Nystr\"{o}m method was by design, in an attempt to weaken the effectiveness of the split function. To understand the effect of finite sampling error on the SVM model produced, we employ Lemma \ref{lemma:k_k_tilde_main} to show that we can bound error in the model output -- proof in \ref{sec:error_bounds}.

\begin{lemma}\label{lemma:pred_error_fin_samp}
Let $f(\cdot)=\sum_i \alpha_i k(\cdot, x_i)$ be the ideal Nystr\"{o}m approximated model and $\widetilde{f}(\cdot)=\sum_i \alpha_i' \widetilde{k}(\cdot, x_i)$ be the equivalent perturbed solution as a result of additive finite sampling error on kernel estimations. With high probability, we can bound the error as,
\begin{equation}
    |f(x)-\widetilde{f}(x)| \leq \mathcal{O}_{M}\left( \frac{N^{4/3}\sqrt{L}}{\sqrt{M}} \right) 
\end{equation}
\noindent where $\mathcal{O}_M$ expresses the term hardest to suppress by $M$.
\end{lemma}
\noindent Lemma \ref{lemma:pred_error_fin_samp} indicates that $M\sim \mathcal{O}(N^3L)$ will suffice to suppress errors from sampling. This is in comparison to $M\sim \mathcal{O}(N^4)$ for without the Nystr\"{o}m approximation \cite{10.1038/s41567-021-01287-z}. However, note that this is to approximate the linear function $f$ for which the SVM has some robustness with only  $\mathrm{sign} [f(\cdot)]$ \eqref{eq:split_function} required.
A caveat to also consider is that it is not necessarily true that a smaller $|| K - \widetilde{K} ||_2$ or $|f(x)-\widetilde{f}(x)| $ will give greater model performance  \cite{sharp_anal_kern_matr_approx}. 
 
To address the complexity of the QRF model, for an error of $\epsilon = \mathcal{O}(1/\sqrt{M})$ on kernel estimations, we can show that training requires $\mathcal{O}(TL(d-1)N\epsilon^{-2})$ samples and single instance prediction complexity of $\mathcal{O}(TL(d-1)\epsilon^{-2})$ circuit samples -- discussed further in Appendix \ref{section:complexity_appendix}. The main result here is that we are no longer required to estimate $\mathcal{O}(N^2)$ elements, as is the case for QSVMs. Though this is a profound reduction for larger datasets with $L<<N$, it should be noted that datasets may require $L=\mathcal{O}(N)$. Nonetheless, the number of estimations will never be greater than $N^2$, on the basis that kernel estimations are stored in memory across trees. 

Finally, to show that the QRF contains hypotheses unlearnable by both classical learners and linear quantum models, we extend the concept class generated with the \textit{discrete logarithm problem} (DLP) in \cite{10.1038/s41567-021-01287-z}, from a single dimensional clustering problem to one in two dimensions. We construct concepts that separate the $2$D log-space (torus) into four regions that can not be differentiated by a single hyperplane. Hence, the class is unlearnable by linear QNNs and quantum kernel machines -- even with an appropriate DLP quantum feature map. Classical learners are also unable to learn such concepts due to the assumed hardness of the DLP problem. Since QDTs essentially create multiple hyerplanes in feature space, there exists $f\in\mathbb{H}_{\mathrm{QDT}}$ to emulate a concept in this class. Further details are presented in Appendix \ref{section:qrf_advantage}.

\section{Numerical Results \& Discussion}\label{section:results}

\begin{figure*}
    \centering
    \includegraphics[width=500pt]{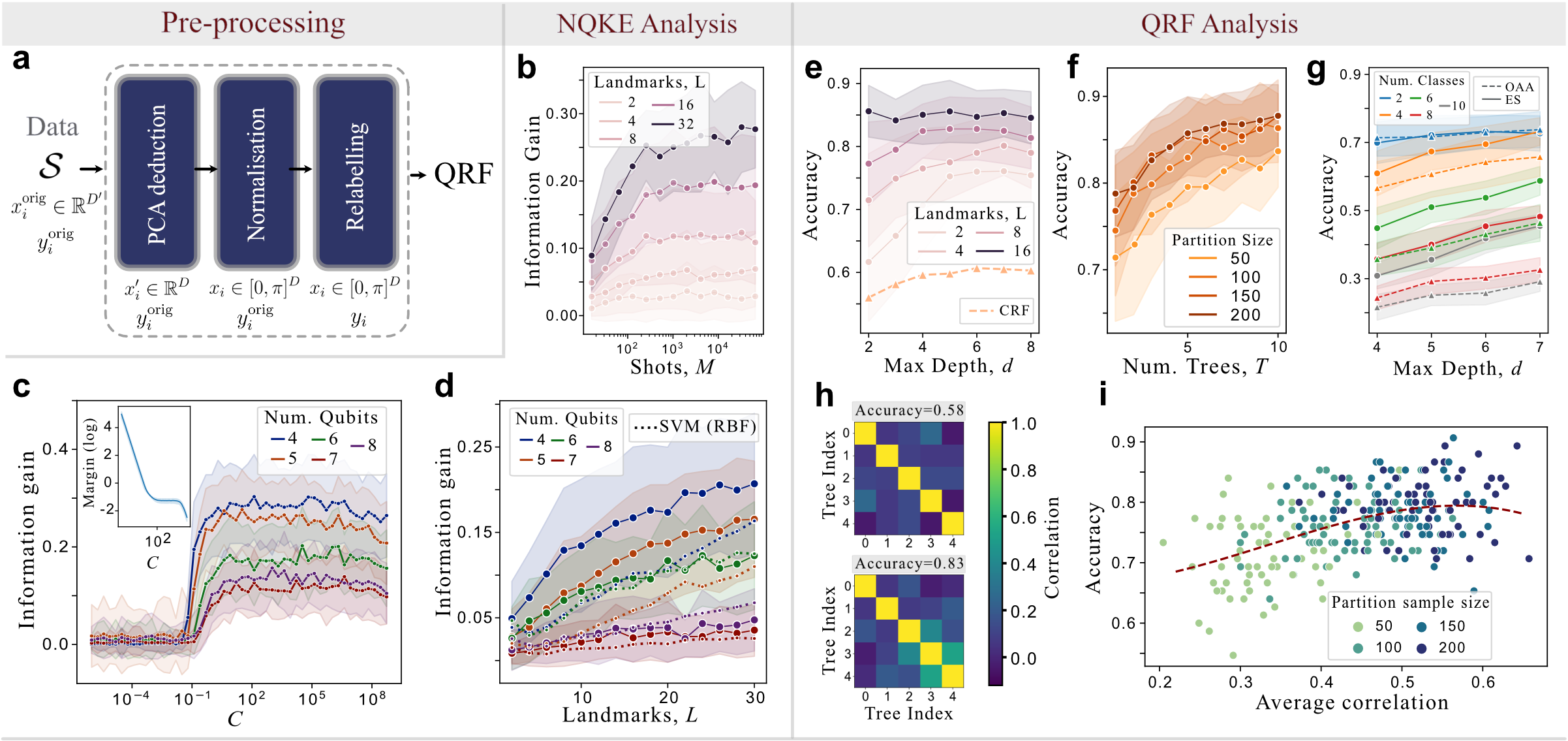}
    \vspace{-0.15cm}
    \caption{\textbf{(a)} Data pre-preprocessing stages preceding model analysis. See Appendix \ref{section:data_and_relabel_process} for further information on the relabelling process. All numerical results in this Figure use the $\mathcal{D}_{\text{FM}}$ dataset with a $3:2$ training/testing split over $300$ randomly sampled instances over a quantum kernel with $\Phi_{\text{IQP}}$ embedding (Appendix \ref{section:data_and_relabel_process}). The embedding is over $7$ qubits, unless otherwise stated. Figures \textbf{(b)} to \textbf{(d)} are the analysis of NQKE split function hyperparameters with each point an average over $50$ trials. \textbf{(b)} has a fixed $C=10$, while \textbf{(c)} has $(M, L)=(2048, 10)$, and \textbf{(c)} has $(M,C)=(2048, 10)$ fixed. The Figures \textbf{(e)} through \textbf{(i)} are the results over the entire QRF model with each point being averaged across 40 independent trials. 
    Figure \textbf{(e)} with fixed parameters $(M,T, C)=(2048,1, 10)$, indicates with large enough $L$, the tree structure has marginal benefit. This is largely a consequence of relabelling instances to be optimal for the quantum kernel. Figure \textbf{(f)} illustrates the effect of increasing the number of trees in the ensemble, for a range of partition sizes -- i.e. the number of randomly sampled instances given to each tree. This is for fixed parameters $(M, L, d, C) = (2048, 5, 4, 10)$. \textbf{(g)} Performance of the QRF on multi-class problems with the original data labels of $\mathcal{D}_{\text{FM}}$. The exact classes introduced for each multi-class problem is given in Appendix \ref{section:data_and_relabel_process}. These results have the following fixed parameters: $(T, M, L, C)=(1, 2048, 5, 10)$. Finally we have the correlation results of Figures \textbf{(h)} and \textbf{(i)}, with \textbf{(h)} displaying a heatmap of the correlation between trees, and \textbf{(i)} showing the average correlation between trees in a QRF model plotted against its accuracy on test instances. These have fixed parameters, $(T, M, L, d, C) = (10, 2048, 5, 4, 10)$.}
    \label{fig:tree_params_figure}
\end{figure*}

The broad structure of the QRF allows for models to be designed specifically for certain problems. This includes the selection of a set of embeddings at each level of the tree that distinguish data points based on different aspects of the real-world dataset.
However, this brings about the requirement of selecting hyperparameters that are not necessarily trivial to determine as optimal. 
There is often an interplay between hyperparameters that can mutually influence their effectiveness. 

In this Section we provide numerical results (with all results simulated with shot noise) of the performance of the model when hyperparameters are varied. Importantly, this is preceded by the classical pre-processing of data -- illustrated in Figure \ref{fig:tree_params_figure}a -- which involves performing \textit{principal component analysis} (PCA) to adjust the dimensionality of a dataset to an embedding on a particular number of qubits. This is followed by normalisation $x_i\in[0,\pi]$ and a relabelling step. This work explores three types of relabelling strategies: (i) utilising the original labels of the dataset, $y_i=y_i^\mathrm{orig}$, (ii) relabelling dataset so that the geometric difference between the classical \textit{radial basis function} kernel and the quantum kernel with respect to a specific embedding are maximised \cite{10.1038/s41467-021-22539-9}, $(y_i)_i = \mathscr{R}_{\Phi}^{\text{QK}}(\{(x_i, y_i^\mathrm{orig})\}_i)$, and, (iii) relabelling the dataset so that they form discontinuous regions in quantum feature space, $(y_i)_i = \mathscr{R}_{\Phi}^{\text{QRF}}(\{(x_i, y_i^\mathrm{orig})\}_i)$, where the label $\text{QRF}$ indicates that such a relabelling process was constructed to highlight the performance of the QRF over other models. Further elaboration of relabelling can be found in Appendix \ref{section:data_and_relabel_process}. 

The effect of varying hyperparameters on the split function are shown in Figures \ref{fig:tree_params_figure}b-d. 
The result of increasing the number of landmark points are evident in \ref{fig:tree_params_figure}b and \ref{fig:tree_params_figure}d, with greater $L$ allowing for a better Nystr\"{o}m approximation. However, the increase in $L$ highlights the importance of bounding the first term of Eq. \eqref{eq:k_k_tilde_bound} as it then requires a larger number of circuit samples to limit the detrimental effects of shot noise. Hence, results exhibit behaviour that were identified theoretically in Lemma \ref{lemma:k_k_tilde_main}. Figure \ref{fig:tree_params_figure}c indicates the appropriate scale of the regularisation parameter, with signs of slight overfitting with larger $C$ favouring accuracy over wider margins. 

The overall performance of the QRF with respect to the hyperparameters of maximum depth $d$, number of trees $T$, number of landmark points $L$, and partition size are shown in Figures \ref{fig:tree_params_figure}e and \ref{fig:tree_params_figure}f. Increasing the depth of the tree is seen to compensate the approximation error introduced by selecting a smaller subset of landmark points. Furthermore, increasing the number of trees in the forest is seen to have a significant impact on the learner's performance, with the number of instances randomly sampled for each tree (partition size) only modestly increasing model accuracy. This suggests no obvious advantage using all available training instances for all trees. Instead, we observe that a bagging approach, with trees receiving subsets of points, is quite effective. 

One of the most distinguishing characteristics of the QRF model is its ability to accommodate multi-class problems. In Figure \ref{fig:tree_params_figure}g we see the performance of a single QDT on multi-class problems with the class split types: OAA and ES. The latter performs far better, with the QRF able to achieve similar accuracy on the four-class problem as with the binary problem. Crucially, this is with no change to the algorithm and further complexity. The leaf nodes simply output a probability distribution over $|\mathcal{C}|$ classes. 
Though we naturally expect the required depth of the tree to increase with $|\mathcal{C}|$, considering the exponential growth of the number of leaf nodes -- with respect to the depth of the tree -- we conjecture a required sub-linear growth in depth, $d=\mathcal{O}(\mathrm{poly}(\log |\mathcal{C}|))$. 

Ensemble machine learning methods are only useful in cases where its constituent weak classifiers are not highly correlated. Though we want individual learners to be correlated to the expected class labels, this must be balanced with the ineffectiveness of highly correlated classifiers. Figure \ref{fig:tree_params_figure}h and \ref{fig:tree_params_figure}i illustrate this by observing the \textit{Spearman's rank correlation} \cite{doi:https://doi.org/10.1002/9781118445112.stat05964} between trees.

The analysis of model hyperparameters now allows us to observe the QRF's performance in comparison to other machine learning models. Figure \ref{fig:model_comparison} shows accuracy of four models: the QRF, the SVM with a quantum kernel (QSVM), the classical \textit{random forest} (CRF) and the SVM with a classical RBF kernel. The comparison is given for embeddings $\Phi_{\text{IQP}}$ and $\Phi_{\text{Eff}}$ (that apply to the two quantum models, see Appendix \ref{section:data_and_relabel_process} for definitions) and the following datasets: $\mathcal{D}_{\text{Io}}$, $\mathcal{D}_{\text{H}}$, $\mathcal{D}_{\text{BC}}$ and $\mathcal{D}_{\text{FM}}$. Further information about these embeddings and datasets can be found in Appendix \ref{section:qrf_construction}. We see that the QRF is able to out-perform the QSVM in all cases where the dataset has been relabelled. This is surprising as $\mathscr{R}_\Phi^{\text{QK}}$ is specifically designed for a single quantum kernel. In these cases the QRF also out-performs the other classical learners. The results on the original class labels are not surprising. In most cases the classical learners have a greater performance, which conforms with results that address the requirement of an embedding with the correct inductive bias \cite{inductive_bias_of_qk}. In other words we require a kernel function that induces an \textit{effective} metric on the dataset. Hence we reach a limitation of quantum kernel methods: there exist no general $\Phi$ for which the kernel function produces a favourable higher dimensional map for classical datasets. We discuss this and other limitations further in Appendix \ref{sec:limitations_q_kernels}. However, it is for this reason we do not attempt to solve the predicament of obtaining useful embeddings, instead emphasising providing proof of concepts through relabeling datasets.

\begin{figure}
    \centering
    \includegraphics[width=240pt]{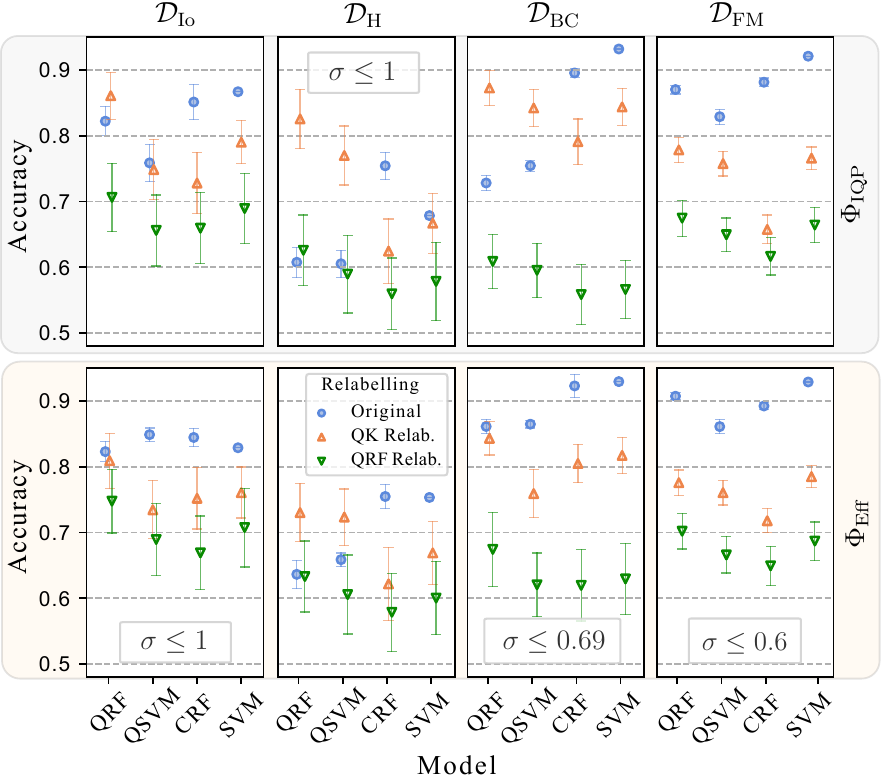}
    \caption{The comparison of quantum and classical classifiers on four datasets: $\mathcal{D}_{\text{Io}}$, $\mathcal{D}_{\text{H}}$, $\mathcal{D}_{\text{BC}}$ and $\mathcal{D}_{\text{FM}}$. The numeric simulations were carried out with two embeddings: $\Phi_{\text{IQP}}$ and $\Phi_{\text{Eff}}$; and the three relabelling strategies: original label, QK relabelled ($\mathscr{R}_\Phi^{\text{QK}}$) and QRF relabelled ($\mathscr{R}_\Phi^{\text{QRF}}$). Note: the two embeddings affect only the quantum models. QRF models have the following hyperparameters, $(T, M, L,d) = (5, 2048, 10, 4)$ for the first three datasets and we increase to $T=10$ for the larger $\mathcal{D}_{\text{FM}}$ dataset. The CRF has identical $(T,d)$ to the QRF. Due to the sensitivity of the models to $C$, the $\Phi_{\text{IQP}}$ model results have $C=10$ for QRF, QSVM and SVM classifiers, while $C=100$ is set for $\Phi_{\text{Eff}}$. Kernel estimations are stored throughout the depth and across the ensemble of trees. Hence, $\sigma$ states the ratio between kernel estimations required by QRF against QSVM.
    Note: $\mathscr{R}_\Phi^{\text{QRF}}$ generates new labels for each trial by randomly selecting new pivot points. This results in the larger error bars shown. The datasets, embeddings and relabelling process are elaborated further in Appendix \ref{section:data_and_relabel_process}.}
    \label{fig:model_comparison}
\end{figure}

\section{Conclusion}\label{section:conclusion}

The well-known \textit{no-free-lunch} theorem in ML \cite{shalev-shwartz_ben-david_2014} refers to the inconvenient fact that there exists no single learning model that is optimal for all datasets. The commonly explored QNN and QSVM methods are linear models in quantum feature space, with the form $f(x)=\Tr[\rho(x)O_\theta]$. 
In this work, we show that the introduction of a learner with a decision tree structure generates a unique quantum learning model that is able to produce hypotheses that can provide a clear advantage due to its non-linearity. Furthermore, the tree structure of the QDT can compensate for a quantum embedding that is unable to perfectly separate instances in feature space. This is demonstrated both, theoretically with the construction of a DLP-based concept class, and numerically with the QRF that is observed to outperform QSVMs in most cases with identical embeddings. 

In addition, unlike many other quantum models, the QRF can produce a probabilistic output of the estimated likelihood of class labels for a particular instance. This further allows for multi-class problems to be addressed with no structural change in the model. Employing a low-rank Nystr\"{o}m approximation (as we know generalisation suffers with large rank \cite{10.1038/s41467-021-22539-9}) to reduce the effectiveness of split functions, reduces the complexity from $\mathcal{O}(N^4)$ to $\mathcal{O}(N^3L)$ shots per estimation required to limit the errors introduced by finite sampling on the linear function representing the hyperplane. However, one should note that the split function takes the sign of the function and therefore such a large number of samples are not required in practice. Moreover, as QDTs are intended to be weak learners there are potential regularisation effects due to such noise \cite{noisy_quantum_kernel_machines}. The Nystr\"{o}m method also allows for the complexity of the number of kernel estimations of the entire model to not grow quadratically with the number of training instances -- which is a significant limitation of QSVMs. Possible improvements to NQKE can observe alternate methods to select landmark points. Examples include the selection of columns by incomplete Cholesky decomposition with column pivoting \cite{10.5555/645529.657980}.

The required model hyperparameters for any given data set and its relationship with the number of qubits used for the quantum split function, is not a trivial question and there are various avenues in which the QRF can be optimised. Though we have numerically explored varying a range of hyperparameters, the QRF model can easily be altered to respond to the specific problem at hand. This includes the use of multiple kernel embeddings down the QDT, or further reducing the complexity of the model for larger datasets by taking only a subset of instances for training at each node -- not only a subset of points to act as landmarks. Furthermore, there is the possibility of utilising techniques such as boosting \cite{boosting_qsvm} and other pruning methods, such as CART's cost complexity pruning algorithm \cite{li1984classification}, to improve the performance of the QRF model. 

\begin{acknowledgments}

MS is supported by the Australian Government Research Training Program (RTP) Scholarship. CDH is partially supported by the Laby Foundation research grant. The authors would like to thank Mario Kieburg for useful discussions. We acknowledge the support provided by the University of Melbourne through the establishment of an IBM Network Quantum Hub. 
The large number of simulations required for this work were made feasible through access to the University of Melbourne's High Performance Computer, \textit{Spartan} \cite{spartan}.

\end{acknowledgments}

\section*{Author Contributions}

MS conceived the QRF approach and carried out the mathematical development, computations and analysis under the supervision of CH and LH. MS wrote the paper with input from all authors.

\section*{Competing Interests}

The Authors declare no competing financial or non-financial interests.

\section*{Data Availability}

Data is available upon reasonable request. The QRF model Python code used to obtain results, can be accessed from the following repository \cite{github}.

\bibliography{My_Library.bib}

\onecolumngrid
\appendix

\section{Background in Machine Learning}\label{section:background}

This work incorporates two well-known machine learning (ML) algorithms: the \textit{support vector machine} (SVM) and the \textit{random forest} (RF). However, before we introduce these methods, we will first identify some of the terminology that will be used throughout this work and are common in ML literature.

We start with the concepts of supervised and unsupervised learning. Supervised models are those that take in for training a set of pairs, $\{(\vec{x}_i, y_i)\}_{i=1}^{N}$ where $\vec{x}_i \in \mathbb{R}^D$ is a $D$-dimensional data vector (also referred to as an \textit{instance}) and $y_i \in \mathcal{Y}$ its associated class label. A binary classification model, with $\mathcal{Y} = \{-1, 1 \}$ labelling the two possible classes, is an example of a supervised method where previously labelled data is used to make predictions about unlabelled instances. Unsupervised models on the other hand, involve finding patterns in data sets that do not have associate labels, $\{\vec{x}_i\}_{i=1}^{N}$. An example is the method of clustering data into groups which looks at finding underlying patterns that may group subsets of instances. This work, however, primarily focuses on the former and hence the remainder of the supplementary document will only refer to supervised algorithms.

The training stage of a model concerns the optimisation of internal model parameters that are algorithmically obtained. However, there are -- in many cases -- \textit{hyperparameters} that must be selected manually prior to training. A common example are regularisation terms that force the model to behave in a certain way. They are used to ensure that the model does not \textit{overfit} to the training data. 
The terms of under- and over-fitting are often used to describe the ways in which a model can fall short from making optimal predictions. Under-fitting occurs when the model does not have the ability to fit to the underlying pattern which may occur if for example there is not enough parameters in model. On the other hand, overfitting is often associated with having an excess of parameters, where the model too closely optimises towards the training data set and is thereby unable generalise its predictions to instances not seen. 

\subsection{Decision Trees and Random Forests}\label{section:rf_model}

In this section, we will give a brief overview of the classical random forest (RF) model and its associated hyperparameters. The classical method closely resembles the QRF proposed in the paper, diverging at the implementation of the split function. Hence it is important that one is familiar with the classical form to understand the implementation of the QRF. The section will start with an introduction to the decision tree (DT) before constructing the RF as the ensemble of DTs. We subsequently discuss the ways in which the RF is enhanced to ensure an uncorrelated ensemble, and therefore its implications to the QRF algorithm. 

\subsubsection{Decision Trees}

Decision trees are supervised ML algorithms that can be employed for both classification and regression purposes. In this work, we focus on the former and give an overview of the implementation.
As discussed in the main text, the DT has a directed tree structure that is traversed to obtain a classification. A data point begins at the root node and chooses a particular branch based on the outcome of a split function, $f$. At this point it is important to identify the differences between the phases of training and prediction. Training starts with a group of instances at the root node while traversing down the nodes training their associated split functions to return the \textit{best} splitting of the data so as to distinguish instances of different class. Prediction on the other hand, simply involves the traversal of the tree that was constructed at training. The number of splits at each node of a DT is usually selected manually and is therefore a hyperparameter of the model. However, two groups is a more natural split for the QRF and we therefore we only observe the case where the DT has a node branching of two. This splitting process continues until a leaf condition is met. These are the three conditions that were defined in the main paper and repeated here for convenience. A node is a leaf node if any of the following apply: (i) instances supplied to the node are of the same class, (ii) the number instances supplied is less than some user-defined value, $m_s$, or (iii) the node is at the maximum depth of the tree, $D$. Condition (i) is clear as further splitting is unnecessary. Condition (ii) ensures that the splitting does not result in ultra-fine splitting with only a few instances at the leaf nodes. Such a result can often indicate that a model has overfitted to the training set. Condition (iii) is in fact a hyperparameter that we will discuss shortly.

In the main text it was seen that goal of the tree was to isolate classes down different branches. Hence, after each step down the tree, the model is more and more certain of the class probability distribution for a particular instance. This is seen in Figure \ref{fig:decision_tree_and_svm}a as we see an amplification of certain classes probabilities down any particular branch of the tree. As an aside, for many this would be a reminder of Grover's algorithm and it is clear as to why some of the very first quantum analogues of DT involved Grover's algorithm \cite{10.1007/s11128-013-0687-5}. 

We saw that each node amplified certain classes through certain branches, but not how this is carried out. There are various methods used in practice to form a split function, resulting in various forms of tree algorithms such as CART, ID3, CHAID, C4.5. 
In this paper we use CART to compare against our quantum model.

\begin{figure}
    \centering
    \includegraphics[width=460pt]{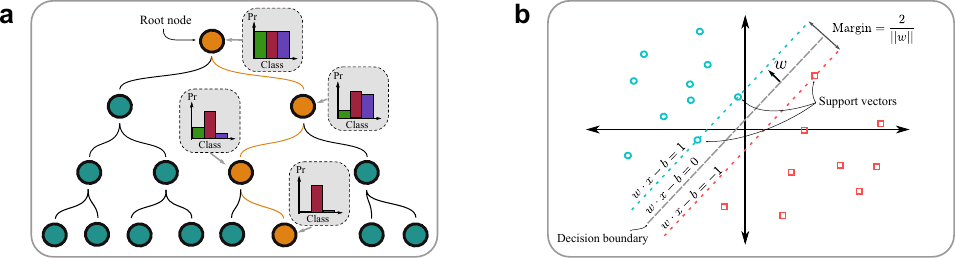}
    \caption{\textbf{(a)} Amplifeication of class probabilities are shown through a particular path through nodes down a Decision Tree. Three classes are depicted with their probability distribution shown as a graph next to each node in the path. During the prediction stage of the model, an instance that concludes at the bottom orange node, will be predicted to be in the maroon class. \textbf{(b)} The support vector machine (SVM) is illustrated with the model optimising for a separating hyperplane with maximum margin. It should be noted that though the SVM, as illustrated above, looks to require linearly separable data, SVMs can employ kernels and slack terms to circumvent this problem. }
    \label{fig:decision_tree_and_svm}
\end{figure}

\vspace{0.15cm}

\paragraph*{\textbf{CART algorithm.}}

Classification and Regression Tree (CART) \cite{li1984classification} is one of the most commonly used DT that is able to support both numerical and categorical target variables. CART constructs a binary tree using the specific feature and threshold that attains a split with the largest information gain at each node. Mathematically, given data $\mathcal{S}_j = \{(x_i, y_i) \}_{i=1}^{N_j}$ at a particular node $j$, the algorithm chooses the optimal parameter $\theta = (l, t)$, where $l$ is a feature and $t$ is the threshold variable, such that the information gain (as defined in \eqref{eq:ig}) is maximised:
\begin{equation}
    \theta^*_j = \argmax_{\theta} \mathrm{IG}\Big(\mathcal{S}_j; \mathcal{S}_j^L (\theta), \mathcal{S}_j^R(\theta) \Big)
\end{equation}
\noindent where $\mathcal{S}_j^L(\theta) $ and $\mathcal{S}_j^R(\theta)$ are the partitioned data sets defined as,
\begin{align}
    \mathcal{S}_j^L(\theta) &=  \{(x,y) \in \mathcal{S}_j | x^{(l)} \leq t\}   \\
    \mathcal{S}_j^R(\theta) &= \mathcal{S}_j \backslash \mathcal{S}_j^L(\theta).
\end{align}
\noindent where $x^{(l)}$ is the $l$th component of vector $x$. This maximisation is then done recursively for $\mathcal{S}_j^L(\theta) $ and $\mathcal{S}_j^R(\theta)$ until a leaf condition is met. Clearly the CART algorithm is most natural for continuous feature spaces with the geometric interpretation of slicing the feature space with hyperplanes perpendicular to feature axes. This can in fact be generalised to oblique hyperplanes -- employed by Perceptron decision trees (PDTs) --  where the condition $x^{(l)} \leq t_j$ becomes $ x\cdot w\leq t_j$, giving an optimisation over $\theta = (w, t_j)$. This optimisation is generally carried out using gradient descent, while it is also possible to use a \textit{support vector machine} (SVM). This is in fact inspiration for the quantum decision tree proposed in this work. The QDT-NQKE developed in Section \ref{section:qrf_main} and elaborated further in Appendix \ref{section:qrf_construction}, obtains an optimal hyperplane with a kernel-SVM employing a quantum kernel. 

Aside from the main approach of the model, there are consequential hyperparameters that greatly affect the structure of tree at training. The first of which has been already set -- that is the number of splits at a node. Since we are not dealing with categorical attributes in this work, setting the number of splitting at two can be compensated by increasing the allowed depth of the tree. This brings us to the next hyperparameter, the \textit{maximum depth}, $d$, of the tree that regulates the activation of criterion (iii) for identifying a leaf node. Allowing a larger depth for the tree means that the DT becomes more expressive while at the same time becoming more prone to overfitting. This is a classic example of the bias-variance trade-off in statistical learning theory that one must consider.

\subsubsection{Random Forests}

The decision tree alone is in fact a weak learner, as it has the tendency to overfit to the trained data. In other words, it struggles with generalising predictions to instances that were not supplied during training. Understanding the problem of overfitting is crucial to all ML algorithms and the use of regularisation techniques are required for the high performance of many models. In the case of the decision tree, overfitting is addressed by taking an ensemble (forest) of trees. An ensemble alone, however, is unlikely to provide relief from this problem. Many trees making identical decisions -- as they are trained from the same data -- will not decrease overfitting. Using an ensemble, one must ensure that each tree classifier is trained uniquely. This is done through injecting randomness. 

It is fundamental that randomness is incorporated into the model during its training. This ensures that the classifiers are uncorrelated and hence provide the most benefit from the ensemble structure. The ways in which randomness is commonly introduced are, bagging, boosting and randomised node optimisation (RNO). Bagging is an attempt to reduce the variance of a model by generating random subsets of the training data set for each decision tree in the forest. Boosting is essentially an extension to this, however the trees are learned sequentially with instances that are poorly predicted occurring with more frequency in subsequent trees. The idea being, instances that are harder to learn are sampled more often -- hence \textit{boosted}. Finally, RNO makes random restrictions on the way that a split function can be optimised. In practice this could be the selection of only a subset of features over which to train the split function. The forest of trees with these random additions are therefore referred to as a \textit{random forest} (RF).
Randomness is at the heart of RFs and hence it is crucial that the QRF is able to inject randomness into its structure to profit from the ensemble created. We will see that this arises naturally for the QRF in Section \ref{section:qrf_construction}.

\subsection{Support Vector Machines}\label{section:svm}

A Support Vector Machine (SVM) is a supervised learning algorithm that aims to find a separating hyperplane (decision boundary) with maximum margin between two classes of instances. Compared to other ML models, SVMs tend to perform well with comparatively small numbers of training instances and become impractical for data sets of more than a few thousand. This gives a regime in which we employ the use of a SVM, when constructing the QRF.

\subsubsection{Linear SVM}\label{section:linear_svm}

We now introduce the SVM starting with the most natural linear case before then showing that the method can be extended to non-linear decision boundaries with the use of kernels. A more in-depth discussion of SVMs can be found in \cite{10.7551/mitpress/4175.001.0001}. Given a set of training instances from a binary concept class, $\{(x_i, y_i) \}_{i=1}^K$ where $(x_i, y_i)\in \mathbb{R}^D \times \{-1, 1\}$, the SVM attempts to find a separating hyperplane that is defined by a vector perpendicular to it, $w\in \mathbb{R}^D$, and a bias term, $b\in \mathbb{R}$. Assuming that such a plane exists (requiring the data to be linearly separable), we have the condition,
\begin{equation}\label{eq:svm_decision_boundary}
    y_i (x_i \cdot w + b) > 0, \forall i=1, ...., K
\end{equation}
\noindent However, here we may observe two details. Firstly, there may exist many such hyperplanes for which this condition is satisfied, and secondly, $w$ remains under-determined in this form, specifically its norm, $||w||$. Both these problems are addressed by further requiring the following inequality,
\begin{equation}\label{eq:svm_condition}
    y_i (x_i \cdot w + b) \geq 1, \forall i=1, ...., K
\end{equation}
\noindent Clearly, Eq. \eqref{eq:svm_condition} implies \eqref{eq:svm_decision_boundary}, however we also have the interpretation of introducing a region on either side of the hyperplane where no data points lie. The plane-perpendicular width of this region is referred to as the \textit{margin} (shown in Figure \ref{fig:decision_tree_and_svm}b) where, from Eq. \eqref{eq:svm_condition}, it can be shown to be $2/||w||_2$. Intuitively we strive to have the largest margin possible and hence we formulate the problem as a constrained convex optimisation problem, 
\begin{align}
    \text{minimise \ \ }& \frac{1}{2} ||w||^2_2 \nonumber \\
    \text{subject to: \ \ }& y_i (x_i \cdot w + b) \geq 1, \forall i=1, ...., K
\end{align}
\noindent We can further introduce \textit{slack} terms to allow the SVM to be optimised in cases with non-separable data sets, 
\begin{align}\label{eq:optim_prob_svm}
    \text{minimise \ \ }& \frac{1}{2} ||w||^2_2 + \frac{\lambda}{p}\sum_{i=1}^{K} \xi_i^{p} \nonumber  \\
    \text{subject to: \ \ }& y_i (x_i \cdot w + b) \geq 1-\xi_i, \forall i=1, ...., K \nonumber \\
    &\xi_i \geq 0,  \forall i=1, ...., K 
\end{align}
\noindent where $\lambda > 0$ is a regularisation term that adjusts the willingness of the model to accept slack, and $p\in \mathbb{N}$ is a constant. Most often, the L2 soft margin problem is solved with $p=2$ resulting in a quadratic program. Hence we can write the \textit{primal} Lagrangian for the L2 soft margin program as,
\begin{equation}\label{eq:primal_lagrangian}
    L(w, b, \xi; \alpha, \mu) = \frac{1}{2} ||w||^2_2 + \frac{\lambda}{2}\sum_{i=1}^{K} \xi_i^{2} + \sum_{i=1}^{K} \Big \{ \alpha_i \Big[ 1 - \xi_i - y_i (x_i \cdot w + b)\Big] \Big\} - \sum_{i=1}^{K} \mu_i \xi_i
\end{equation}
\noindent where $\alpha, \mu \in \mathbb{R}^K$ are \textit{dual variables} of the optimisation problem, also known as \textit{Lagrange multiplier vectors}. In practice this optimisation problem is reformulated into the \textit{Lagrange dual} problem, with the \textit{Lagrange dual function},
 \begin{align}
     \mathcal{L}(\alpha, \mu) &= \inf_{w, b, \xi} L(w, b, \xi; \alpha, \mu)\\
     &= \sum_{i=1}^K \alpha_i - \frac{1}{2}\sum_{i, j}^K \alpha_i \alpha _j y_i y_j x_i \cdot x_j - \frac{1}{2\lambda}\sum_{i=1}^K (\alpha_i + \mu_i)^2 \label{eq:dual_lagrange} \\
     &\text{\ \ \ s.t. \ \ } \sum_{i=1}^K \alpha_i y_i = 0\nonumber 
 \end{align}
\noindent where Eq. \eqref{eq:dual_lagrange} is obtained from setting the partial derivatives of the primal Lagrangian (in Eq. \eqref{eq:primal_lagrangian}), with respect to $w,b,\xi$, to zero. The dual problem is subsequently defined as, 
\begin{align}
    \text{maximise \ \ }& \mathcal{L}(\alpha, \mu) \nonumber \\
    \text{subject to: \ \ }& \alpha_i \geq 0, \forall i=1, ...., K \nonumber \\
    &\mu_i \geq 0, \forall i=1, ...., K
\end{align}
\noindent The dual program is convex (this in fact independent of the convexity of the primal problem) with saddle point optimal solutions. However, with a convex primal problem in \eqref{eq:optim_prob_svm}, the variables $(w^*, b^*, \xi^*, \alpha^*, \mu^*)$ that satisfy the \textit{Karush-Kuhn-Tucker} (KKT) conditions are both primal and dual optimal. Here, the KKT conditions are,
\begin{align}
    w - \sum_{i=1}^K \alpha_i y_i x_i &= 0 \label{eq:kkt_w} \\
    \sum_{i=1}^K \alpha_i y_i &= 0 \\
    \lambda \xi_i - \alpha_i - \mu_i &= 0 \\
    \alpha_i \big[ 1 - \xi_i - y_i (x_i \cdot w + b)\big] &= 0 \label{eq:kkt_support_vec}\\
      y_i (x_i \cdot w + b) - 1 +  \xi_i &\geq 0 \\
      \mu_i \xi_i &= 0 \\
      \mu_i &\geq 0 \\
      \xi_i &\geq 0 
\end{align}
\noindent Solving the KKT conditions amount to solving the problem of obtaining an optimal hyperplane. It is clear that the expression of the hyperplane in Eq. \eqref{eq:kkt_w} allows us to write the classification function for the linear SVM as, 
\begin{equation}\label{eq:h(x)_for_svm}
    h(x) = \mathrm{sign} \left( \sum_{i=1}^N \alpha_i y_i x\cdot x_i + b \right)
\end{equation}
\noindent We can however make a further adjustment by realising that Eq. \eqref{eq:kkt_support_vec} implies, only data points that lie on the margin have non-zero $\alpha$. These points are referred to as \textit{support vectors} and are illustrated in Figure \ref{fig:decision_tree_and_svm}b. Hence, defining the index set of support vectors, $\mathcal{S} \subseteq \{1, ..., N\}$, we have, 
\begin{equation}
    h(x) = \mathrm{sign} \left( \sum_{s\in \mathcal{S}} \alpha_s y_s x\cdot x_s + b \right)
\end{equation}

\noindent As a final note, the dual formulation allowed us to formulate the problem with training instances present only as pairwise similarity comparisons, i.e. the dot product. This allows one to generalise linear SVMs to non-linear SVMs by instead supplying an inner product on a transformed space.

\subsubsection{Non-linear SVM using kernels}

Non-linear decision boundaries can be trained using an SVM by drawing a linear hyperplane through data embedded in a non-linearly transformed space.
Let $\phi:\mathcal{X}\xrightarrow{}\mathcal{H}$ be a \textit{feature map} such that, $x\xrightarrow{}\phi(x)$, where $\mathcal{H}$ is a Hilbert space with inner product $\langle \cdot, \cdot \rangle_{\mathcal{H}}$. Hence, we have a \textit{kernel} function defined by the following, 
\begin{definition}\label{def:kernel_def}
Let $\mathcal{X}$ be a non-empty set, a function $k:\mathcal{X} \times \mathcal{X} \xrightarrow{} \mathbb{C}$ is called a \textit{kernel} if there exists a $\mathbb{C}$-Hilbert space and a map $\phi:\mathcal{X}\xrightarrow{}\mathcal{H}$ such that $\forall x_i, x_j \in \mathcal{X}$, 
\begin{equation}\label{eq:kernel_def}
    k(x_i, x_j)=\langle \phi(x_i), \phi(x_j) \rangle_{\mathcal{H}}
\end{equation}
\end{definition}
\noindent Though we intend to use quantum states that live in a complex space, the kernels used in this work will in fact map onto the real numbers. This will be made clear through the following elementary theorems \cite{paulsen_raghupathi_2016}.
\begin{theorem}\label{theorem:sum_kernels}
    (Sum of kernels are kernels). Let $k_1, k_2:\mathcal{X} \times \mathcal{X} \xrightarrow{} \mathbb{C}$ be kernels. Then $$k(x_i, x_j) := k_1(x_1, x_2) + k_2(x_1, x_2)$$ for $x_i, x_j \in \mathcal{X}$ defines a kernel. 
\end{theorem}
\begin{theorem}\label{theorem:prod_kernels}
(Product of kernels are kernels). Let $k_1, k_2:\mathcal{X} \times \mathcal{X} \xrightarrow{} \mathbb{C}$ be kernels. Then $$k(x_i, x_j) := k_1(x_1, x_2) \cdot k_2(x_1, x_2)$$ for $x_i, x_j \in \mathcal{X}$ defines a kernel. 
\end{theorem}
\noindent The proofs for these theorems are quite straightforward and can be found in \cite{10.1007/978-0-387-77242-4}. These theorems allow one to construct more complex kernels from simpler ones, and as alluded to earlier, they allow us to understand the transformation from an inner product on the complex field to one on the real. Let $k':\mathcal{X}\times\mathcal{X}\xrightarrow{}\mathbb{C}$ and $k'':=k'^*$ be kernels, where $^*$ denotes the complex conjugate. From Theorem \ref{theorem:prod_kernels} we are able to define a kernel, $k: \mathcal{X}\times\mathcal{X} \xrightarrow{} \mathbb{R}$ such that $k(x_i, x_j) = k'(x_1, x_2) \cdot k''(x_1, x_2) = |k'(x_1, x_2)|^2$. In Section \ref{section:qrf_construction} we will see that this is how we construct the quantum kernel: as the fidelity between two feature embedded quantum states, $k(x_i, x_j) = |\langle \phi(x_j) | \phi(x_j) \rangle|^2$. 

The kernel function is clearly conjugate symmetric due to the axioms of an inner product. It is important to note that given a feature map, the kernel is unique. However, the reverse does not hold true: a kernel does not have a \textit{unique} feature map. Furthermore, we will see that one does not need to define a feature map to be able to claim having a kernel. For this we need the concept of \textit{positive definteness}.
\begin{definition} 
    A symmetric function, $k: \mathcal{X} \times \mathcal{X} \xrightarrow{} \mathbb{R}$ is \textit{positive definite} if for all $a_i \in \mathbb{R}$, $x_i \in \mathcal{X}$ and $N\geq 1$,
    \begin{equation}
        \sum_{i=1}^N \sum_{j=1}^N a_i a_j k(x_i, x_j) \geq 0
    \end{equation}
\end{definition}
\noindent All inner products are positive definite, which clearly implies the positive definiteness of the kernel in Eq. \eqref{eq:kernel_def}. Interestingly, this is in fact an equivalence, with a positive definite function guaranteed to be an inner product on a Hilbert space $\mathcal{H}$ with an appropriate map $\phi:\mathcal{X} \xrightarrow{} \mathcal{H}$ \cite{10.1007/978-0-387-77242-4}.  
\begin{theorem}
(Symmetric, positive definite functions are kernels). A function $k:\mathcal{X} \times \mathcal{X} \xrightarrow{} \mathbb{R}$ is a kernel if and only if it is symmetric and positive definite. 
\end{theorem}
\noindent This theorem allows us to generate kernels without requiring the specific feature map that generated the kernel. Furthermore, we are able to compute a kernel that may have a feature map that is computationally infeasible. This concept is crucial to the emergence of kernel methods and will be the main inspiration for \textit{quantum kernel estimation}. 

Returning to the issue of non-linear decision boundaries, we observe that the linear SVM, from Section \ref{section:linear_svm}, can be extended by simply generalising the inner product on Euclidean space with a kernel function. This is known as the \textit{kernel trick}, where data is now represented as pairwise similarity comparisons that may be associated with an embedding in a higher dimensional (or even infinite dimensional) space. We therefore have the following dual program,
\begin{align}
    \text{maximise \ \ }& \sum_{i=1}^N \alpha_i - \frac{1}{2}\sum_{i, j}^N \alpha_i \alpha _j y_i y_j k(x_j, x_j) - \frac{1}{2\lambda}\sum_{i=1}^N (\alpha_i + \mu_i)^2 \label{eq:svm_kernel_program} \\
    \text{subject to: \ \ }& \alpha_i \geq 0, \forall i=1, ...., K  \\
    &\mu_i \geq 0, \forall i=1, ...., N  \\
     &\sum_{i=1}^K \alpha_i y_i = 0 
\end{align}
\noindent The classification function is also quite similar with the replacement of the kernel for the dot product. However, we will obtain its form from the discussion of \textit{reproducing kernels} and their associated \textit{reproducing kernel Hilbert spaces}.

\subsubsection{Reproducing Kernel Hilbert Spaces and the Representer Theorem}\label{sec:representer_theorem}

The concept of Reproducing Kernel Hilbert spaces (RKHS) are invaluable to the field of statistical learning theory, as they accommodate the well known \textit{Representer theorem}. This will give us a new perspective on the derivation of the SVM optimisation. Furthermore it will provide a guarantee that an optimal classification function will be composed of weighted sums of kernel functionals over the training set.
\begin{definition}\label{def:rk_and_rkhs}
Let $\mathcal{F}\subset \mathbb{C}^{\mathcal{X}}$ be a set of functions forming a Hilbert space with inner product $\langle \cdot, \cdot \rangle_{\mathcal{F}}$ and norm $||f||=\langle f, f\rangle_{\mathcal{F}}^{\frac{1}{2}}$ where $f\in\mathcal{F}$. A function $\kappa:\mathcal{X} \times\mathcal{X} \xrightarrow{} \mathbb{F}$, for some field $\mathbb{F}$, is called a \textit{reproducing kernel} of $\mathcal{F}$ provided that,
\begin{enumerate}[label=(\roman*), leftmargin=3cm]
\item $\forall x \in \mathcal{X}$, $\kappa_x (\cdot) := \kappa(\cdot, x) \in \mathcal{F}$, and,
\item $\forall x \in \mathcal{X}$, $\forall f\in\mathcal{F}$, $\langle f, \kappa(\cdot, x) \rangle_{\mathcal{F}} = f(x)$ (reproducing property)
\end{enumerate}
\noindent are satisfied. The Hilbert space, $\mathcal{F}$ -- for which there exists such a reproducing kernel -- is therefore referred to as a reproducing kernel Hilbert Space.
\end{definition}
\noindent It is important to note that the definition of a kernel is not explicitly stated in Definition \ref{def:rk_and_rkhs}. Rather, we see that for any $x_i,x_j\in \mathcal{X}$, we have the following property of the reproducing kernel:
\begin{equation}
    \kappa(x_i, x_j) = \langle \kappa(\cdot, x_i), \kappa(\cdot, x_j) \rangle_\mathcal{H}
\end{equation}
\noindent When compared with Eq. \eqref{eq:kernel_def}, we see that we have a feature map of the form, $\phi(x) = \kappa(\cdot, x)$. Therefore it is evident that the reproducing property implies that $\kappa$ is a kernel as per Definition \ref{def:kernel_def}. We also have the reverse:
\begin{theorem}
(Moore-Aronszajn \cite{10.1090/s0002-9947-1950-0051437-7}). For every non-empty set $\mathcal{X}$, a function $k:\mathcal{X}\times \mathcal{X} \xrightarrow{}\mathbb{R}$ is positive definite if and only if it is a reproducing kernel. 
\end{theorem}
\noindent This implies that every positive definite kernel is associated with a \textit{unique} RKHS. This has implications for the Nystr\"{o}m approximated quantum kernel. However, we leave reconstructing the associated RKHS for future work. 

To understand the significance of the representer theorem, we state the general problem kernel methods attempt to solve -- this includes the non-linear SVM. Given some loss function $\mathscr{L}$ that quantifies the error of a learning model against labelled training data $\{(\vec{x}_i, y_i)\}_{i=1}^{N}$, we aim to find an optimal function $f^*$ such that, 
\begin{equation}\label{eq:arbitrary_loss_svm_optim}
  f^* =  \argmin_{f\in \mathcal{F}} \frac{1}{N} \sum_{i=1}^K \mathscr{L}(y_i, f(x_i)) + \lambda ||f||^2_{\mathcal{F}}
\end{equation}
\noindent where $\lambda\geq 0$ and $\mathcal{F}$ is the RKHS with reproducing kernel $k$. This optimisation problem is quite general and further encompasses algorithms such as \textit{kernel ridge regression} \cite{vovk2013kernel}. However, it is not clear that such an optimisation is efficiently computable. We will see that the \textit{representer theorem} will allow us to simplify the problem so that the whole space of functions need not be searched in order to find the optimal $f^*$. 
\begin{theorem}
(The Representer Theorem \cite{kimeldorf1971some, 10.1007/3-540-44581-1_27}). Let $\mathcal{F}$ be a RKHS with associated reproducing kernel, $k$. Given a set of labelled points $\{(\vec{x}_i, y_i)\}_{i=1}^{N}$ with loss function $\mathscr{L}$ and a strictly monotonically increasing regularisation function, $P:\mathbb{R}^{+}_0 \xrightarrow{} \mathbb{R}$, we consider the following optimisation problem,
\begin{equation}\label{eq:rep_theorem1}
    \min_{f\in \mathcal{F}} \frac{1}{N} \sum_{i=1}^N \mathscr{L}(y_i, f(x_i)) + P(||f||_{\mathcal{F}})
\end{equation}
Any function $f^* \in \mathcal{F}$ that minimises \eqref{eq:rep_theorem1}, can be written as, 
\begin{equation}
    f^* = \sum_{i=1}^N \alpha_i k(\cdot, x_i)
\end{equation}
\noindent where $\alpha_i \in \mathbb{R}$. 
\end{theorem}
\noindent The significance of this theorem comes from the fact that solutions to kernel methods with high dimensional (or even infinite) functionals, are restricted to a subspace spanned by the representers of the data. Thereby reducing the optimisation of functionals to optimising scalar coefficients.

Now coming back to SVMs, the optimal classification function, $f^*$ is the solution of the following, 
\begin{equation}\label{eq:svm_opt_w}
    W^* = \argmin_{W\in \mathcal{F}} \frac{1}{N} \sum_{i=1}^N \max(0, W(x_i)y_i) + \frac{\lambda}{2} ||W||^2_{\mathcal{F}}
\end{equation}
\noindent where we will see $W^*:\mathcal{X} \xrightarrow{} \mathbb{R}$ corresponds to the optimal separating hyperplane.
Now using the representer theorem, we let $W = \sum_{i=1}^N \alpha_i k(\cdot, x_i)$ and substitute into Eq. \eqref{eq:svm_opt_w} to obtain the program derived in \eqref{eq:svm_kernel_program}. Though the approach to the optimisation of the SVM in this section is different to the maximisation of the hyperplane margin illustrated earlier, we are required to solve the same convex optimisation problem.  

\subsection{Generalisation error bounds}\label{section:gen_err_bound_clas}

The most important characteristic of a machine learning model is that it is able to \textit{generalise} current observations to predict outcome values for previously unseen data. This is quantified though, what is known as, the \textit{generalisation error} of a model.
\begin{definition}\label{def:gen_error}
\textit{(Generalisation error)} Given a hypothesis $h\in \mathbb{H}$ a target concept $c\in \mathcal{C}$ and an underlying distribution $D$, the generalisation error (or \textit{risk}) of $h$ is defined as, 
\begin{equation}
    \mathcal{R}(h) = \Pr_{x\sim \mathcal{D}} [h(x) \neq c(x)]
\end{equation}
\end{definition}
\noindent However, both the underlying distribution $\mathcal{D}$ of the data and the target concept $c$ are not known, and hence we have \textit{empirical error} across the testing set.
\begin{definition}
\textit{(Empirical Error)} Given a hypothesis $h\in \mathbb{H}$ and samples $S=(z_1 , ..., z_N)$, where $z_i=(x_i, y_i)$, the \textit{empirical error} of $h$ is defined as,
\begin{equation}
    \widehat{\mathcal{R}}(h) = \frac{1}{N}\sum_{i=1}^N \mathbf{1}_{h(x_i)\neq y_i}
\end{equation}
where $\mathbf{1}_a$ is the indication function of the event $a$, $1$ when $a$ is true and $0$ otherwise.
\end{definition}
The aim of this section is to provide theoretical bounds on the generalisation error, with the the empirical error being depicted in numerical results. To provide these bounds, we must first distinguish the strength of models. There exist a range of tools in statistical learning theory to quantify the richness of models. One such measure is the \textit{Rademacher complexity} that measures the degree to which a model -- defined by its hypothesis set -- can fit random noise. Note that this is independent of the trainability of the model.
\begin{definition}
\textit{(Empirical Rademacher complexity)} Given a loss function $\Gamma:\mathcal{Y}\times \mathcal{Y}\xrightarrow{}\mathbb{R}$ and a hypothesis set $\mathbb{H}$, let $\mathcal{S} =((x_1, y_1) , ..., (x_N, y_N))$ be a fixed sample set of size $N$, and $\mathcal{G}=\{g:\mathcal{X}\times \mathcal{Y}\xrightarrow{}\mathbb{R}| g(x,y)=\Gamma(h(x), y), h\in\mathbb{H}\}$ be a family of functions. The empirical Rademacher complexity of $\mathcal{G}$ with respect to the sample set $\mathcal{S}$ is defined as, 
\begin{equation}\label{eq:emp_rade_complexity}
    \widetilde{\mathfrak{R}}_{\mathcal{S}} (\mathcal{G}) = \mathop{\mathbb{E}}_{\sigma} \Bigg[ \mathop{\mathrm{sup}}_{g\in \mathcal{G}}\frac{1}{N} \sum_{i=1}^N \sigma_i g(x_i, y_i)\Bigg]
\end{equation}
where $\sigma_i \in \{-1, +1\}$ are independent uniform random variables known as \textit{Rademacher variables}.
\end{definition}
\begin{definition}
\textit{(Rademacher complexity)} Let samples be drawn from some underlying distribution $\mathcal{D}$, such that $\mathcal{S}\sim \mathcal{D}^N$. For $N\geq 1$, the Rademacher complexity of $\mathcal{G}$ is defined as the expectation of the empirical Rademacher complexity over all samples of size $N$, each drawn from $\mathcal{D}$, i.e.,
\begin{equation}
    \mathfrak{R}_N (\mathcal{G}) = \mathop{\mathbb{E}}_{\mathcal{S}\sim \mathcal{D}^N} \big[\widetilde{\mathfrak{R}}_{\mathcal{S}} (\mathcal{G}) \big]
\end{equation}
\end{definition}
\noindent The Rademacher complexity now allows us to state a well known theorem in statistical learning theory that provides an upper bound to the expectation value of function $g\in \mathcal{G}$.
\begin{theorem}\label{theorem:rade_bound}
(Theorem 3.1, \cite{10.1007/s00362-019-01124-9}) Let $\mathcal{G}$ be a family of functions mapping from $\mathcal{X}\times \mathcal{Y}$ to $[0,1]$. Then for any $\delta >0$ and for all $g\in \mathcal{G}$, with probability of at least $1-\delta$ we have,
\begin{equation}
     \mathbb{E}_{(x,y)\sim\mathcal{D}}[g(x, y)] \leq \frac{1}{N} \sum_{i=1}^N g(x_i, y_i) + 2 \widetilde{\mathfrak{R}}_{\mathcal{S}} (\mathcal{G}) + 3 \sqrt{\frac{\log(1/\delta)}{2N}}
\end{equation}
\end{theorem}
\noindent We see from Theorem \ref{theorem:rade_bound} that a model that is more expressive, i.e. has a greater $\widetilde{\mathfrak{R}}_{\mathcal{S}}$, is in fact detrimental to bounding the loss on unseen data points. This supports the observation of overfitting to training instances that most often occur when a model is over-parameterised. 

The theory of generalisation error discussed in this section will become important when we attempt to quantify the performance of the QRF and its constituent parts. It is however important to note that the Rademacher complexity is only one particular way of attaining a generalisation bound. We will later come across sharper bounds that can be obtained using \textit{margin}-based methods in the context of separating hyperplanes.

\subsection{Nystr\"{o}m method for kernel-based learning}\label{section:nystrom_intro}

The elegance of kernel methods are greatly limited by the $\mathcal{O}(N^2)$ operations required to compute and store the kernel matrix, where $N$ is the number of training instances. In the age of \textit{big data} this greatly reduces the practicality of kernel methods for many large-scale applications. 
Furthermore, the number of kernel element computations become crucial in Quantum Kernel Estimation, where quantum resources are far more limited. 
However, there have been solutions to reduce this complexity that have had varying degrees of success \cite{10.1038/s41534-021-00498-9}. In this work we approximate the kernel matrix using the Nystr\"{o}m method that requires only a subset of the matrix to be sampled. This approximation is predicated on the \textit{manifold hypothesis} that in statistical learning refers to the phenomena of high dimensional real-world data sets usually lying in low dimensional manifolds. The result is that the rank of the kernel matrix is often far smaller than $N$, in which case an approximation using a subset of data points is a plausible, effective solution.

The Nystr\"{o}m method was introduced in the context of integral equations with the form of an eigen-equation \cite{using_nystr_speed_kern},
\begin{equation}\label{eq:int_eignfn_eq}
    \int k(y, x) \phi_i (x) p(x) dx = \lambda_i \phi_i (y)
\end{equation}
\noindent where $p(x)$ is the probability density function of the input $x$, $k$ is the symmetric positive semi-definite kernel function, $\{\lambda_1, \lambda_2, ...\}$ denote the non-increasing eigenvalues of the associated $p$-orthogonal eigenvectors $\{\phi_1, \phi_2, ...\}$, i.e. $\int \phi_i(x)\phi_j(x)p(x)dx=\delta_{ij}$ . Equation \eqref{eq:int_eignfn_eq} can be approximated given i.i.d. samples $\mathcal{X} = \{x_j\}_{j=1}^N$ from $p(x)$ by replacing the integral with the sum as the empirical average,
\begin{equation}
    \frac{1}{N} \sum_{j=1}^N k(y, x_j) \phi_i (x_j) \approx \lambda_i \phi_i (y)
\end{equation}
This has the form of a matrix eigenvalue problem, 
\begin{equation}\label{eq:mat_eignfn_eq}
    K^{(N)} \Phi^{(N)} = \Phi^{(N)} \Lambda^{(N)}
\end{equation}
\noindent where $K_{ij}^{(N)}=k(x_i, x_j)$ for $i,j=1,...,N$ is the kernel (Gram) matrix, $\Phi_{ij}^{(N)} \approx \frac{1}{\sqrt{N}}\phi_j(x_i)$ containing the eigenvectors and the diagonal matrix $\Lambda_{ii}^{(N)} \approx N\lambda_i$ containing the eigenvalues. Therefore, solving this matrix equation will subsequently give approximations to the initially sought-after eigenvector $\phi_i$ in Eq. \eqref{eq:int_eignfn_eq}, 
\begin{equation}
    \phi_i(y) \approx \frac{\sqrt{N}}{\Lambda_{ii}^{(N)}} \sum_{j=1}^N k(y, x_j) \Phi^{(N)}_{ji}
\end{equation}
The effectiveness of this approximation is determined by the number of samples, $\mathcal{X}=\{x_i\}_{i=1}^N$ -- with a greater number $q$ resulting in a better approximation. Applying similar reasoning, a subset $\mathcal{Z}=\{z_i\}_{i=1}^{L} \subset \mathcal{X}$ with $L< N$ points can be used to approximate the eigenvalue problem in Eq. \eqref{eq:mat_eignfn_eq}. This approximation is precisely the Nystr\"{o}m method, with the points in set $\mathcal{Z}$ referred to as \textit{landmark} points. More specifically, given the eigen-system of the full kernel matrix, $K \Phi_{\mathcal{X}} = \Phi_{\mathcal{X}} \Lambda_{\mathcal{X}}$ with $K_{ij}=k(x_i, x_j)$ and equivalently $W \Phi_{\mathcal{Z}} = \Phi_{\mathcal{Z}} \Lambda_{\mathcal{Z}}$ with $W_{ij}=k(z_i, z_j)$, we make the following approximation \cite{using_nystr_speed_kern}, 
\begin{equation}\label{eq:ny_approx_phi}
    \Phi_{\mathcal{X}} \approx \sqrt{\frac{L}{N}} E \Phi_{\mathcal{Z}}\Lambda_{\mathcal{Z}}^{-1} \text{ , \ \ } \Lambda_{\mathcal{X}} \approx \frac{N}{L} \Lambda_{\mathcal{Z}}
\end{equation}
\noindent where $E\in \mathbb{R}^{N\times L}$ with $E_{ij} = k(x_i, z_j)$ assuming without loss of generality $E:=[W,B]^\top$, $B\in \mathbb{R}^{(N-L)\times L}$. Combining $K= \Phi_{\mathcal{X}} \Lambda_{\mathcal{X}} \Phi_{\mathcal{X}}^\top$ with the approximation in Eq. \eqref{eq:ny_approx_phi} we find, 
\begin{align}
    K &\approx \Bigg(\sqrt{\frac{L}{N}} E \Phi_{\mathcal{Z}}\Lambda_{\mathcal{Z}}^{-1}\Bigg) \Big( \frac{N}{L} \Lambda_{\mathcal{Z}} \Big)
    \Bigg( \sqrt{\frac{L}{N}} E \Phi_{\mathcal{Z}}\Lambda_{\mathcal{Z}}^{-1}\Bigg)^\top \\
    &= EW^{-1}E^\top\label{eq:ny_kernel_approx}
\end{align}
where $W^{-1}$ is the pseudo-inverse of $W$. In practice, one usually takes the best $r$-rank approximation of $W$ with respect to the spectral or Frobenius norm prior to taking the pseudo-inverse, i.e. $W^{-1}_r = \sum_{i=1}^r \lambda_i^{-1} u_i u_i^\top$ where $r\leq L$ with orthonormal eigenvectors $\{u_i\}_{i=1}^L$ and associated non-increasing eigenvalues $\{\lambda_i\}_{i=1}^L$ of matrix $W$. This will later ensure that we can avoid problems of $\norm{W^{-1}}_2$ becoming unbounded as arbitrarily small eigenvalues of $W$ are omitted. Now expanding matrix $E$ in Eq. \eqref{eq:ny_kernel_approx} we have the Nystr\"{o}m approximation given by, 
\begin{equation}\label{eq:nystrom_kernel}
    K \approx \widehat{K} :=
    \begin{bmatrix}
    W      &  B \\
    B^\top      & B^\top W^{-1} B  
    \end{bmatrix}.
\end{equation}
In other words, by computing only the $L$ columns $E=[W, B]^\top\in \mathbb{R}^{N\times L}$ we are able to approximate the full $K \in \mathbb{R}^{N\times N}$ Gram matrix required for use in kernel methods. In practice, for its application in SVMs we are able to interpret this approximation as a map. To see this, we expand the approximation in Eq. \eqref{eq:ny_kernel_approx},  
\begin{align}
    \widehat{\Phi}^\top \widehat{\Phi} := \widehat{K} &= EW^{-1}E^\top \\
    &= E W^{-\frac{1}{2}} W^{-\frac{1}{2}} E^\top \\
    &=\Big(E W^{-\frac{1}{2}} \Big) \Big(E W^{-\frac{1}{2}} \Big)^\top
\end{align}
where we use the fact that $(W^{-\frac{1}{2}})^\top =  W^{-\frac{1}{2}}$. Hence, we have an approximated kernel space, $\widehat{\Phi}=\Big(E W^{-\frac{1}{2}} \Big)^\top$, i.e., given a point $x$, the associated vector in kernel space is  $\widehat{x}=(vW^{-1/2})^\top\in \mathbb{R}^{L}$ where $v=[k(x, z_1), ..., k(x, z_{L})]$. To make this clear, we explicitly define the \textit{Nystr\"{o}m feature map} (NFM), $\mathbf{N}_W:\mathcal{X}\xrightarrow{}\mathbb{R}^{L}$, such that, 
\begin{equation}\label{eq:nystrom_feature_map}
    \Big(\mathbf{N}_W (x)\Big)_i = \sum_{j=1}^{L} k(x, z_j) (W^{-1/2})_{ij} 
\end{equation}
\noindent where $i=1, ...,L$ and $\{z_i\}_{i=1}^{L}$ are the landmark points of the approximation. The set $\widehat{\mathcal{X}}=\{ \mathbf{N}_W (x) : x \in \mathcal{X} \}$ is now used as the training set for fitting a linear SVM. Therefore, employing Nystr\"{o}m approximation for kernel SVMs, is identical to transforming the training data to the kernel space prior to the training of the SVM. 

It should be noted that the performance of the approximation is significantly influenced by the specific landmark points $\{z_i\}_{i=1}^{L}$ selected -- especially since we generally have $L<<N$. Consequently, there has been a great deal of work exploring strategies to sample sets of optimal landmark points. Apart from uniform sampling from the data set, it was suggested in \cite{JMLR:v6:drineas05a} that the $i$th column be sampled with weight proportional to its diagonal element, $k(x_i, x_i)$. In the case of quantum kernel estimation, where the diagonal entries are always $1$, both techniques are therefore identical. Furthermore there exist more complex schemes based on using $k$-means clustering to sample columns \cite{10.1109/tnn.2010.2064786} and greedy algorithms to sample based on the feature space distance between a potential column and the span of previously chosen columns. In this work we randomly sample columns and it is left for future work whether the exist any advantages to do otherwise with regards to the quantum case.

Computationally, the Nystr\"{o}m method requires performing a singular value decomposition (SVD) on $W$ that has a complexity of $\mathcal{O}(L^3)$. In addition to the matrix multiplication required, the overall computational complexity is $\mathcal{O}(L^3 + NL^2)$ while the space complexity is only $\mathcal{O}(L^2)$. For the case where $L<<N$ this is a significant improvement to the $\mathcal{O}(N^2)$ space and time complexity required without Nystr\"{o}m. Nonetheless, this method is only an approximation and hence it is useful to provide bounds to its error.  

\subsubsection{Error bounds to low-rank nystr\"{o}m approximation}

The Nystr\"{o}m method is most effective in cases where the full kernel matrix $K$ has low-rank. In such cases the approximation (i) does not depend heavily on the specific columns chosen, and (ii) is able to faithfully project to a lower dimensional space to form a low-rank matrix that does not differ greatly from the real matrix. However, the rank of a specific kernel matrix is dependent on both the training set as well as the specific kernel function, meaning that a rank dependent bound would be impractical. We therefore provide bounds on the error of the Nystr\"{o}m approximation through the spectral norm of difference between approximated and actual kernel matrices, $\norm{K-\widehat{K}}_2$.
It should be noted that, empirically, it is generally observed that as you increase the rank, the relative prediction error decays far more quickly than the error in matrix approximation \cite{sharp_anal_kern_matr_approx}. This suggests that bounds supplied based on the matrix norms may not necessarily be optimal if the goal is maximising the overall performance of the model. Nonetheless, we provide the following bound to the Nystr\"{o}m approximation error that will become important for bounding the generalisation error of our quantum model.

\begin{theorem}\label{theorem:bound_classical_nystrom}
(Theorem 3, \cite{JMLR:v6:drineas05a}) Let $K$ be an $N\times N$ symmetric positive semi-definite matrix, with Nystr\"{o}m approximation $\widehat{K}=EW^{-1}E^\top$ constructed by sampling $L$ columns of $K$ with probabilities $\{p_i\}_{i=1}^N$ such that $p_i =K_{ii}/\sum_{j=1}^N K_{jj}$. In addition, let $\epsilon>0$ and $\eta=1+\sqrt{8\log(1/\delta)}$. If $L\geq 4\eta^2/\epsilon^2$ then with probability at least $1-\delta$, 
\begin{equation}
    \norm{K-\widehat{K}}_2 \leq \epsilon \sum_{i=1}^N K_{ii}^2
\end{equation}
\end{theorem}
\noindent Following this theorem, we have the subsequent Corollary that bounds the error in the kernel matrix in the case where the diagonal matrix entries are equal to $1$. This special case is true for quantum kernels. 
\begin{corollary}\label{corollary:bound_classical_nystrom}
Given $K_{ii}=1$ for all $i=1,...,N$, with probability at least $1-\delta$ we have,
\begin{equation}\label{eq:bound_classical_nystrom}
    \norm{K-\widehat{K}}_2 \leq \frac{N}{\sqrt{L}}\Big( 1+ \sqrt{8 \log (1/\delta)}\Big) = \mathcal{O}\Bigg(\frac{N}{\sqrt{L}}\Bigg).
\end{equation}
\end{corollary}
\begin{proof}
Follows trivially from choosing $L=4\eta^2 / \epsilon^2$.
\end{proof}

The bounds in Theorem \ref{theorem:bound_classical_nystrom} and Corollary \ref{corollary:bound_classical_nystrom} are clearly not sharp. There have been subsequent works providing sharper bounds based on the spectrum properties of the kernel matrix, such as in \cite{10.1109/tit.2013.2271378} where the authors improve the upper bound to $\mathcal{O}(N/L^{1-p})$, where $p$ is the $p$-power law of the spectrum. In other words, they show that the bound can be improved by observing the decay of the eigenvalues, with a faster decay giving a stronger bound.

\section{The Quantum Random Forest}\label{section:qrf_construction}

This section will elaborate upon the inner workings of the QRF beyond the overview discussed in the main text. This includes the algorithmic structure of the QRF training process, the quantum embedding, datasets used and relabelling strategies employed. We also supply proofs for the error bounds claimed in the main text.

The QRF is an ensemble method that averages the hypotheses of many DTs. To obtain an intuition for the performance of the averaged ensemble hypothesis, we observe the $L2$ risk in the binary learning case and define the the ensemble hypothesis as $H := \mathbb{E}_{\mathbf{Q}} [h_{\mathbf{Q}}]$, where $h_{\mathbf{Q}}$ is the hypothesis for each DT. We observe the following theorem.
\begin{theorem}\label{theorem:risk_ensemble_classifiers}
(\textit{$L2$-risk for averaged hypotheses}). Consider $\mathbf{Q}$ as a random variable to which a hypothesis $h_{\mathbf{Q}}: \mathcal{X} \xrightarrow{} \{-1,+1\}$ is defined. Defining the averaged hypothesis as $H := \mathbb{E}_{\mathbf{Q}} [h_{\mathbf{Q}}]$ and $\mathrm{Var}[h_{\mathbf{Q}}] := \mathbb{E}_{\mathbf{Q}}[(h_{\mathbf{Q}}(x) - H(x))^2]$, the $L2$ risk, $R$, satisfies,
\begin{equation}\label{eq:risk_ensemble_classifiers}
    R(H) = \mathbb{E}_{\mathbf{Q}} [R(h_{\mathbf{Q}})] - \mathbb{E}_{X}[\mathrm{Var}[h_{\mathbf{Q}}(X)]]
\end{equation}
\end{theorem}
\noindent \textit{Proof.} Derivation can be found on pg. 62 of \cite{sup_learn_notes}.

The theorem highlights the fact that highly correlated trees do not take advantage of the other trees in the forest with the overall model performing similarly to having a single classifier. Alternatively, obtaining an ensemble of trees that are only slightly correlated -- hence with larger $\mathrm{Var}[h_{\mathbf{Q}}] $ -- is seen to give a smaller error on unseen data (risk). Note, we are also not looking for uncorrelated classifiers, as this would imply that the classifiers are always contradicting one another, resulting in the first term on the RHS of Eq. \eqref{eq:risk_ensemble_classifiers} becoming large.
Interestingly, one can also observe the phenomena expressed in Theorem \ref{theorem:risk_ensemble_classifiers} by observing the variance of $\mathbb{E}_{\mathbf{Q}} [h_{\mathbf{Q}}]=\frac{1}{T}\sum_{t=1}^T  \mathbf{Q}_t(\vec{x}; c)$ for $T$ trees and obtaining a result that is explicitly dependent on the correlation between classifiers.

\subsection{Algorithmic structure of training the QRF}\label{section:alg_structure_qrf}

\begin{algorithm}[H]
\caption{Training of the $i^{\text{th}}$ split node}\label{alg:split_fn_train}
\begin{algorithmic}
\Require Pre-processed dataset $\mathcal{S}^{(i)}=\{(x_j, y_j)\}_{j=1}^{N^{(i)}}$ where $N^{(i)} > 0$, $x_i\in [0, \pi]^D$, and an untrained split function $\mathcal{N}^{(i)}_{\Phi, \theta}:\mathbb{R}^D \xrightarrow{}\{-1, +1\}$ defined through embedding $\Phi$ with tunable parameters $\theta$ pertaining to this node. Embedding induces a quantum kernel function $k_\Phi (x_l, x_j) = |\langle \Phi(x_j) | \Phi(x_l)\rangle|^2$ that will be used to optimise $\theta$. $L^{(i)}$ and $C$ are also given as hyperparameters. 
\Ensure Disjoint sets $\mathcal{S}_-^{(i)} = \{ ({x}, y) \in \mathcal{S}^{(i)} | \mathcal{N}^{(i)}_{\theta}({x}) = -1 \}$ and $\mathcal{S}_+^{(i)} = \{ ({x}, y) \in \mathcal{S}^{(i)} | \mathcal{N}^{(i)}_{\theta}({x}) = 1 \}$ with an optimised $\theta$.
\State $\mathrm{ig} \gets 0$
\While{$\mathrm{ig} \leq \delta$}\Comment{In this work we set $\delta=0$}
\State Randomly select $L^{(i)}$ points from $\mathcal{S}^{(i)}$ as \textit{landmarks}, giving the subset $\mathcal{S}^{(i)}_L \subseteq \mathcal{S}^{(i)}$
\State Construct Nystr\"{o}m map, $\mathbf{N}$, from these chosen landmark points -- see Section \ref{section:nystrom_intro}.
\State $x_j' \gets \mathbf{N}(x_j)$ for all $x_j \in \mathcal{S}^{(i)}|_\mathcal{X}$.
\State $y_j' \gets F(y_j)$ for all $y_j \in \mathcal{S}^{(i)}|_\mathcal{Y}$.
\State Train Linear SVM (regularisation parameter $C$) with transformed data points $\mathcal{S}'^{(i)}=\{(x_j', y_j)\}_{j=1}^{N^{(i)}}$. Giving optimal $\theta$.
\State Compute $\mathcal{S}_-^{(i)} = \{ (x', y) \in \mathcal{S}'^{(i)} | \mathcal{N}^{(i)}_{\theta}(x') = -1 \}$ and $\mathcal{S}^{(i)}_+ = \{ (x', y) \in \mathcal{S}'^{(i)} | \mathcal{N}^{(i)}_{\theta}(\vec{x}) = 1 \}$
\State $\mathrm{ig} \gets \mathrm{IG}(\mathcal{S}^{(i)}; \mathcal{S}_-^{(i)}, \mathcal{S}_+^{(i)})$
\State $C\gets C*10$. \Comment{We increase $C$ if $\mathrm{ig} \leq \delta$}
\EndWhile 
\end{algorithmic}
\end{algorithm}

In this Section we provide a step-by-step description of the QRF training algorithm that should complement the explanation given in the main text. For simplicity, we break the training algorithm into two parts: (i) the sub-algorithm of split function optimisation, and, (ii) the overall QRF training. We state the former in Algorithm \ref{alg:split_fn_train}.  

The implementation of the SVM was using the \textit{sklearn} Python library that allowed for a custom kernel as an input function. Having laid out the steps for the optimisation of a split node, we can finally describe the QRF training algorithm in Algorithm \ref{alg:qrf_train}. This will incorporate Algorithm \ref{alg:split_fn_train} as a sub-routine. Example code for the QRF model can be found at the following repository \cite{github}.

\subsection{Embedding}\label{section:embedding}

In this section we talk about the specifics of the PQC architectures used for embedding classical data vectors into a quantum feature space. The importance of appropriate embeddings cannot be understated when dealing with quantum kernels. This is evident when observing the bound on the generalisation error of quantum kernels that is seen in \cite{10.1038/s41467-021-22539-9}. However, as we focus on relabelled instances in this work -- discussed further in Section \ref{section:relab_qk} -- we do not compare embeddings and their relative performance on various datasets. Instead the embeddings provide a check that numeric observations are observed, even with a change in embedding.

This work looks at two commonly analysed embeddings: 
\begin{enumerate}
\item[(i)] \textit{Instantaneous-Quantum-Polynomial} (IQP) inspired embedding, denoted as $\Phi_{\mathrm{IQP}}$. This embedding was proposed in \cite{10.1038/s41586-019-0980-2} and is commonly employed as an embedding that is expected to be classically intractable.
\begin{equation}
    \ket{\Phi_{\mathrm{IQP}} (x_i)} = \mathcal{U}_Z (x_i) H^{\otimes n} \mathcal{U}_Z (x_i) H^{\otimes n} \ket{0}^{\otimes n} 
\end{equation}

\begin{algorithm}[H]
\caption{Training of the QRF classifier}\label{alg:qrf_train}
\begin{algorithmic}
\Require Pre-processed dataset $\mathcal{S}=\{(x_i, y_i)\}_{i=1}^{N}$ where $N > 0$, $x_i\in [0, \pi]^D$, and tuples of $L=(L^{(i)})_{i=1}^d$ and $\Phi=(\Phi^{(i)})_{i=1}^d$
\Ensure A trained QRF model.
\Procedure{TrainQDT}{$\mathcal{S}'$, depth, $L$, $\Phi$}
\If{(depth $\geq d$) or ($|\mathcal{S}'|\leq m_s$) or (all instances of $\mathcal{S}'$ are of the same class)}
\State \Return{Leaf node with class distribution of $\mathcal{S}'$}
\ElsIf{$|\mathcal{S}'|=0$}\Comment{This can only occur if ig is allowed to be zero}
\State \Return{Leaf node with the class distribution of its parent}
\Else
\State Run Algorithm \ref{alg:split_fn_train} with $(\mathcal{S}', L^{\text{(depth)}}, \Phi^{\text{(depth)}})$, obtain trained split function $\mathcal{N}_{\theta}$ and receive subsets $\mathcal{S}_-'$ and $\mathcal{S}_+'$.
\State Node $\gets$ Split node with trained split function $\mathcal{N}_{\theta}$
\State Left child of Node $\gets$ \Call{TrainQDT}{$\mathcal{S}_-'$, depth$+1$, $L$, $\Phi$}
\State Right child of Node $\gets$ \Call{TrainQDT}{$\mathcal{S}_-'$, depth$+1$, $L$, $\Phi$}
\State \Return{Node}
\EndIf
\EndProcedure
\State \emph{Start:}
\State From dataset $\mathcal{S}$, sample a subset for each of the $T$ trees: $\{\mathcal{S}^t\}_{t=1}^T$ where $\mathcal{S}^t \subseteq \mathcal{S}$ \Comment{Bagging approach}
\For{$t=1,...,T$, indexing individual trees}
\State Root node $\gets$ $\Call{TrainQDT}{\mathcal{S}^t, 1, L, \Phi}$
\State Store Root node as the $t^{\text{th}}$ tree in the QRF ensemble
\EndFor
\end{algorithmic}
\end{algorithm}

where $H^{\otimes n}$ is the application of Hadamard gates on all qubits in parallel, and $\mathcal{U}_Z (x_i)$ defined as, 
\begin{equation}
    \mathcal{U}_Z (x_i) = \exp \left( \sum_{j=1}^n x_{i,j}Z_j + \sum_{j=1}^n \sum_{k=1}^n  x_{i,j}x_{i,k} Z_j Z_k \right)
\end{equation}
\noindent where $x_{i,j}$ denotes the $j^{\mathrm{th}}$ element of the vector $x_i \in \mathbb{R}^{D}$. This form of the IQP embedding is identical to that of \cite{10.1038/s41467-021-22539-9} and equivalent to the form presented in \cite{10.1038/s41586-019-0980-2}. One should note, in this work each feature in the feature space is mean-centred and has a standard deviation of $1$ -- see Section \ref{section:data_preproc} for the pre-processing of datasets.
\item[(ii)] \textit{Hardware-Efficient-Anzatz} (HEA) style embedding, denoted as  $\Phi_{\mathrm{Eff}}$, respectively. The anzatz is of the form of alternating transverse rotation layers and entangling layers -- where the entangling layers are CNOT gates applied on nearest-neighbour qubits. Mathematically, it has the form,
\begin{align}
    \ket{\Phi_{\mathrm{Eff}} (x_i)} =& \nonumber\\ \prod_{l=0}^{\mathbf{L-1}}& \Bigg(\mathrm{E}_n \mathcal{R}^Z_n\Big(\{x_{i,k} | k = 2nl + n + j \text{ mod } D\}_{j=0}^{n-1}\Big)  \mathrm{E}_n \mathcal{R}^Y_n\Big(\{x_{i,k} | k = 2nl + j \text{ mod } D\}_{j=0}^{n-1}\Big)  \Bigg)\ket{0}^{\otimes n} 
\end{align}
where $\mathcal{R}^A_n (\{\theta_i\}_{i=0}^{n-1}) = \prod_{i=0}^{n-1}\exp (-i A_j \theta_i/2)$ with $A_j$ defined as Hermitian operator $A$ applied to qubit $j$, defining the entanglement layer as $\mathrm{E}_n = \prod_{j=1}^{\lfloor n/2 \rfloor} \mathrm{CNOT}_{2j,2j+1} \prod_{j=1}^{\lfloor n/2 \rfloor} \mathrm{CNOT}_{2j-1,2j} $, $D$ is the dimension of the feature space and we define $\prod_{i=1}^N A_i := A_N ... A_1$ . In this work we take $\mathbf{L}=n$ such that the number of layers scales with the number of qubits. Though this embedding may seem quite complex, the form in which classical vectors are embedded into  quantum state is quite trivial, as we are simply sequentially filling all rotational parameters in the anzatz with each entry of the vector $x_i$ -- repeating entries when required. 
\end{enumerate}

\begin{figure}
    \centering
    \hspace*{-10pt}\includegraphics[width=427pt]{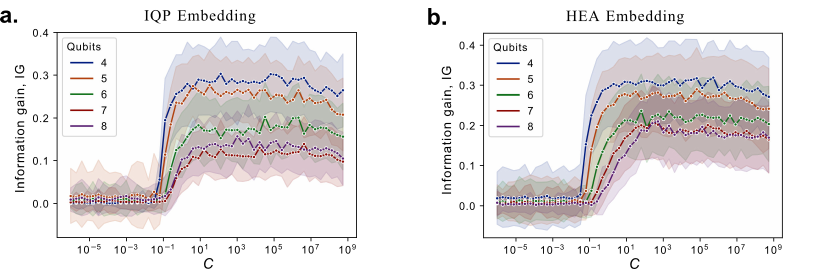}
    \vspace{-0.2cm}
    \caption{The importance of selecting an appropriate penalty parameter, $C$ is shown above. One should note $C\propto 1/\lambda$ from Eq. \eqref{eq:optim_prob_svm}, with a larger $C$ preferring greater accuracy over larger margins. The original \textit{Fashion MNIST} dataset was used with split function parameters $(T, M, L)=(1, 2048, 10)$. Training was done sampling $180$ instances and testing with a separate $120$ instances. These numeric simulations indicate that the dimension of the dataset does also affect where optimal $C$ occurs. It is therefore crucial that $C$ is optimised for any given problem. }
    \label{fig:svm_c_diff_embed}
\end{figure}

\noindent Though both embeddings induce a kernel function that outputs a value in the range $[0,1]$, the regularisation term $C$ can be quite sensitive to the embedding as it is generally a function of the type of embedding (including number of qubits) and the number of data points used for training. This is illustrated in Figure \ref{fig:svm_c_diff_embed} where we see that $\Phi_{\text{Eff}}$ requires a higher $C$.
We posit that the increase of $C$ down the tree allows for wider margins at the top of the tree -- avoiding overfitting -- with finer details distinguished towards the leaves. Demonstration of this is left for future work.

\subsection{Multiclass classification}\label{sec:multiclass_classification}

A crucial decision must be made when dealing with multi-class problems. The complication arises in designating labels to train the binary split function at each node. More specifically, given the original set of classes $\mathcal{C}$, we must determine an approach to find optimal partitions, $\mathcal{C}^{(i)}_{-1}\cup \mathcal{C}^{(i)}_{+1} = \mathcal{C}^{(i)}$ where $\mathcal{C}^{(i)}_{-1} \cap \mathcal{C}^{(i)}_{+1}=\emptyset$ and $\mathcal{C}^{(i)}$ indicates the set of classes present at node $i$. These sets are referred to as \textit{pseudo class partitions} as they define the class labels for a particular split function. The two pseudo class partitions will therefore define the branch down which instances will flow. We consider two strategies for obtaining partitions of classes given a set of classes $\mathcal{C}^{(i)}$ at node $i$:
\begin{enumerate}
    \item[(i)] \textit{One-against-all} (OAA): Randomly select $c \in \mathcal{C}^{(i)}$ and define the partitions to be  $\mathcal{C}_{-1}^{(i)}= \{c\}$ and $\mathcal{C}_{+1}^{(i)}= \mathcal{C}^{(i)}\backslash \{c\}$.  
    \item[(ii)] \textit{Even-split} (ES): Randomly construct pseudo class $\mathcal{C}_{-1}^{(i)}=\{ c_i\}_{i=1}^{\lceil |\mathcal{C}^{(i)}|/2 \rceil} \subset \mathcal{C}^{(i)}$ and subsequently define $\mathcal{C}_{+1}^{(i)}= \mathcal{C}^{(i)}\backslash \mathcal{C}_{-1}^{(i)}$. This has the interpretation of splitting the class set into two. 
\end{enumerate}
Though both methods are valid, the latter is numerically observed to have superior performance -- seen in Figure \ref{fig:multiclass_comp}. The two strategies are identical when $|\mathcal{C}|\leq 3$ and the performance diverges as the size of the set of class increases. We can therefore introduce associated \textit{class maps}, $F_{\mathrm{OAA}}^{(i)}$ and $F_{\mathrm{ES}}^{(i)}$ for a node $i$, defined as,
\begin{equation}
    y \xrightarrow{} \widehat{y} = F^{(i)}(y) := 
     \begin{cases}
    -1,& \text{if } y\in\mathcal{C}_{-1}^{(i)}\\
    +1,              & \text{if } y\in\mathcal{C}_{+1}^{(i)}
\end{cases}
\end{equation}
\noindent where $\widehat{y}$ is the \textit{pseudo class label} for the true label $y$. The distinction between $F_{\mathrm{OAA}}^{(i)}$ and $F_{\mathrm{ES}}^{(i)}$ occurs with the definition of the pseudo class partitions.
\begin{figure}
    \centering
    \hspace*{-9pt}\includegraphics[width=490pt]{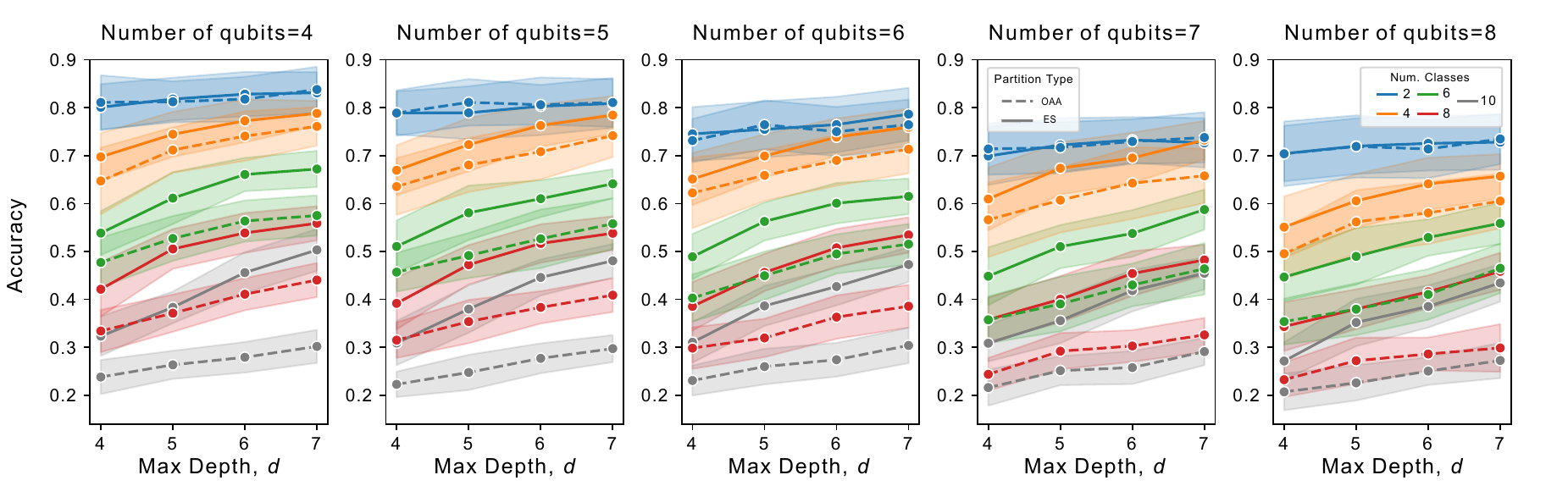}
    \vspace{-0.2cm}
    \caption{Illustrating the difference in performance between the two malticlass partitioning strategies: \textit{one-against-all} (OAA) and \textit{even-split} (ES). The original \textit{Fashion MNIST} dataset was used with QRF parameters $(T, M)=(1, 1024)$. We see that by increasing the maximum depth of the tree, the QDT has the ability to improve the model on multiclass problems, with the ES partition strategy easily outperforming OAA. It should be stressed that these results are over a single tree -- not an ensemble.}
    \label{fig:multiclass_comp}
\end{figure}

\subsection{Complexity analysis}\label{section:complexity_appendix}

One of the motivations of the QRF is to overcome the quadratic complexity with regard to training instances faced by traditional quantum kernel estimation.
In this section we discuss the complexities of the model quoted in the main text of the paper. We start with the quantum circuit sampling complexity -- or simply the sampling complexity. For kernel elements, $K_{ij}$, that are sampled as Bernoulli trials, to obtain an error of $\epsilon$ requires $\mathcal{O}(\epsilon^{-2})$ shots. We then must account for the $N\times L$ landmark columns that need to be estimated to obtain the Gram matrix for a single split node. Since each instance traverses at most $d-1$ split nodes over $T$ trees, we can give an upper bound sampling training complexity of $\mathcal{O}\left(TL(d-1)N\epsilon^{-2}\right)$. Testing over a single instance has a complexity of $\mathcal{O}\left(TL(d-1)\epsilon^{-2}\right)$ as the Nystr\"{o}m transformation requires the inner product of the test instance with the $L$ landmark points. Comparatively, the regular QKE algorithm requires sampling complexity of $\mathcal{O}(N^2\epsilon^{-2})$ and $\mathcal{O}(S\epsilon^{-2})$ for training and testing respectively, where $S$ is the number of support vectors. In \cite{10.1038/s41534-021-00498-9}, it was found that a majority of the training points in SVMs with quantum kernels, persisted as support vectors. Though this is largely dependent on the embedding and specific dataset, a similar characteristic was also observed, therefore allowing us to approximate $\mathcal{O}(N)=\mathcal{O}(S)$. 

It should be noted that the expensive training complexity provided for the QRF are upper bound. This is a result of being able to partition larger datasets into different trees. For example, partitioning $N$ points into $T$ trees therefore gives a complexity of $\mathcal{O}\left(L(d-1)N\epsilon^{-2}\right)$ for training -- thereby, in effect, removing the $T$ dependence. Furthermore, it is not uncommon for leaves to appear higher up the tree, not requiring instances to traverse the entire depth of the tree, with an effective depth, $\mathbb{R}^{+}\ni d_\mathrm{Eff}\leq d$. Finally, since there are at most $N^2$ inner products between the $N$ instances, the sampling complexity of the QRF is always smaller than that of QKE with $\mathcal{O}(N^2)$. 

Though quantum kernel estimations are currently the most expensive, there is also the underlying classical optimisation required to train the SVM. Though there are improvements, solving the quadratic problem of \eqref{eq:svm_kernel_program} requires the inversion of the kernel matrix which has complexity of $\mathcal{O}(N^3)$. In comparison, though not true with the implementation of this work (see Section \ref{section:alg_structure_qrf}), the Nystr\"{o}m method has a complexity of $\mathcal{O}(L^3 + L^2 N)$ \cite{10.1109/tnnls.2014.2359798}. This optimisation is then required at each \textit{split node} in a QDT. For large datasets and a small number of landmark points, this is a substantial improvement -- however this is a general discussion where SVMs are infeasible for large-scale problems.

\subsection{Error bounds}\label{sec:error_bounds}

In this section we start by providing proof on the generalisation error bound of the QRF model by drawing its similarity to the well-studied classical perceptron decision tree. We then walk through the proof of the error bound given in main text for the Nystr\"{o}m approximated kernel matrix. This will finally allow us to bound the error of the SVM-NQKE due to finite sampling -- giving an indication of the complexity of circuit samples required.

 \subsubsection{Proof of Lemma \ref{lemma:gen_err_qdt}: Generalisation error bounds on the QDT model}\label{section:gen_err_qrf}

To see the importance of margin maximisation we draw a connection between the QRF and the Perceptron Decision Tree (PDT) that instead employs perceptrons for split functions. The difference lies in the minimisation problem, with roughly $||w||^2 /2 + \lambda\sum_i [1-y_i(x_i \cdot w + b)]$ minimised for SVMs as opposed to only $\sum_i [1-y_i \sigma(x_i \cdot w + b)]$ (where $\sigma$ is some activation function) for perceptrons. Regardless, the split function remains identical, $f(x) = \mathrm{sign}(x \cdot w + b)$, with optimised $w,b$. As a result, we are able to state the following theorem that applies to both QRFs and PDTs.
\begin{theorem}\label{theorem:gen_err_percep_decision_tree}
(Generalisation error of Perceptron Decision trees, Theorem 3.10 in \cite{10.1023/a:1007600130808}). Let $H_J$ be the set of all PDTs composed of $J$ decision nodes. Suppose we have a PDT, $h\in H_J$, with geometric margins $\gamma_1, \gamma_2, ..., \gamma_J$ associated to each node. Given that $m$ instances are correctly classified, with probability greater than $1-\delta$ we bound the generalisation error,
\begin{equation}\label{eq:gen_err_percep_decision_tree}
    \mathcal{R}(h) \leq \frac{130 r^2}{m} \Bigg( D' \log (4me) \log(4m) + \log \frac{(4m)^{J+1} {2J \choose J}}{(J+1) \delta} \Bigg)
\end{equation}
where $D'=\sum_{i=1}^J 1/\gamma_i^2$ and $r$ is the radius of a sphere containing the support of the underlying distribution of $x$. 
\end{theorem}
\noindent This theorem is obtained from an analysis of the \textit{fat-shattering} dimension \cite{10.1007/978-3-319-21852-6_6} of hypotheses of the form $\{f: f(x)=[\langle w, x\rangle + b>0]\}$.  See \cite{10.1023/a:1007600130808} for a proof of this theorem.
Now, using Stirling's approximation on Eq. \eqref{eq:gen_err_percep_decision_tree},  
\begin{equation}
 \mathcal{R}(h)\leq \widetilde{\mathcal{O}}\left(\frac{r^2}{m} \left[\log(4m)^2 \sum_{i=1}^J \gamma_i^{-2} + J \log\left(4mJ^2\right)\right]\right)   
\end{equation}
\noindent where $\widetilde{\mathcal{O}}(\cdot)$ hides the additional $\log$ terms.
In the context of the QRF model both $x$ and $w$ lie in $\mathcal{H}$ which implies the normalisation condition dictates $||x|| = 1 =: r$. In Section \ref{section:svm} it was seen that $1/||w||$ is proportional to the geometric margin for an optimised SVM (note that this is not true for the PDT in Theorem \ref{theorem:gen_err_percep_decision_tree}), allowing us to write,
\begin{equation}
 \mathcal{R}(h) \leq \widetilde{\mathcal{O}}\left(\frac{1}{m} \left[\log(4m)^2 \sum_{i=1}^J ||w^{(i)}||^2 + J \log\left(4mJ^2\right)\right]\right)
\end{equation}
\noindent We are now required to obtain a bound for $||w||$ that incorporates the kernel matrix. Using the definition of $w$ from Eq. \eqref{eq:kkt_w} we can write $||w||^2 = \sum_{i,j}\alpha_i \alpha_j y_i y_j K_{ij}$. However, the dual parameters $\alpha_i$ do not have a closed form expression and will therefore not provide any further insight in bounding the generalisation error. The equivalent formulation of the SVM problem in terms of the \textit{hinge-loss} is also not helpful as it is non-differentiable. Instead, rather than observing the optimisation strategy of \ref{section:svm}, we make the assumption of minimising over the following mean-squared-error (MSE) which gives a kernel ridge regression formulation,
\begin{equation}\label{eq:ridge_regression}
    \min_{\mathbf{w'}}\lambda ||\mathbf{w'}||^2 + \frac{1}{2}\sum_{i=1}^N \left( \mathbf{w'}^\top \cdot \phi(x_i) - y_i\right)^2
\end{equation}
\noindent where we redefine $w\cdot \phi'(x)+b$ as $\mathbf{w}\cdot \phi(x)$ with $\mathbf{w} := [w^\top, b]^\top$ and $\phi(x):=[\phi(x)'^\top,1]^\top$ to simplify expressions. Importantly, the optimisation problem of Eq. \eqref{eq:ridge_regression} is of the form presented in \eqref{eq:arbitrary_loss_svm_optim}, with the Representer Theorem still applicable. Taking the gradient with respect to $\mathbf{w}$ and solving for zero using the identity $(\lambda I + BA)^{-1}B = B(\lambda I + AB)^{-1}$, we get $\mathbf{w}_{\text{MSE}} = \Phi^\top (K + \lambda I)^{-1} Y$. This gives $||\mathbf{w}_{\text{MSE}}||^2 = Y^\top (K + \lambda I)^{-1} K (K + \lambda I)^{-1} Y$ which is simplified in \cite{10.1038/s41467-021-22539-9, 10.22331/q-2021-08-30-531} taking the regularisation parameter $\lambda\xrightarrow{}0$ which assumes a greater importance for instances to be classified correctly. However, this approximation that gives $||\mathbf{w}_{\text{MSE}}||^2 = Y^\top K^{-1} Y$, has problems of invertibility given a singular kernel matrix -- which in practice is common. Furthermore in the case of the Nystr\"{o}m approximation, we have $\mathrm{det}[\widetilde{K}]=0$ by construction, making a non-zero regularisation parameter crucial to the expression. We therefore replace the inverse with the Moore-Penrose inverse, giving $||\mathbf{w}_{\text{MSE}}||^2 = |Y^\top K^+ Y|$ as $\lambda\xrightarrow{}0$.
Since $1/||w||$ is proportional to the geometric margin and an SVM optimises for a maximum margin hyperplane, we can state that $||\mathbf{w}_{\text{SVM}}|| \leq ||\mathbf{w}_{\text{MSE}}||$ as $\lambda\xrightarrow{}0$, and subsequently conclude the proof with,
\begin{equation}
 \mathcal{R}(h) \leq \widetilde{\mathcal{O}}\left(\frac{1}{m} \left[\log(4m)^2 \sum_{i=1}^J \left|Y^{(i)} (K^{(i)})^+ Y^{(i)\top}\right| + J \log\left(4mJ^2\right)\right]\right),
\end{equation}
\noindent where $K^{(i)}$ and $Y^{(i)}$ are respectively the kernel matrix and labels given to node $i$ in the tree. We therefore arrive at Lemma \ref{lemma:gen_err_qdt} by expanding the matrix multiplication as a sum over matrix elements $(K^{(i)})^+_{jl}$ and labels $\{y^{(i)}_j\}_{j=1}^{N^{(i)}}$. 

One should note that the possibility of a large margin is dependent on the distribution of data instances -- hence losing the distribution-independent bound seen in bounds such as in Theorem \ref{theorem:rade_bound}.
Nonetheless, this theorem is remarkable as the generalisation bound is not dependent on the dimension of the feature space but rather the margins produced. This is in fact common strategy for bounding kernel methods, as there are cases in which kernels represent inner products of vectors in an infinite dimensional space. In the case of quantum kernels, bounds based on the Vapnik-Chervonenkis (VC) dimension would grow exponentially with the number of qubits. The result of Lemma \ref{lemma:gen_err_qdt} therefore suggests that large margins are analogous to working in a lower VC class.

\subsubsection{Proof of Lemma \ref{lemma:k_k_tilde_main}: Error bounds on the difference between optimal and estimated kernel matrices}\label{sec:proof_k_k_tilde}

In the NQKE case, we only have noisy estimates of the kernel elements $W_{ij}\xrightarrow{}\widetilde{W}_{ij}$ and $B_{ij}\xrightarrow{}\widetilde{B}_{ij}$. Furthermore, the Nystr\"{o}m method approximates a large section of the matrix with the final kernel having the form,
\begin{equation}\label{eq:quantum_nystrom_kernel}
    \widetilde{K} :=
    \begin{bmatrix}
    \widetilde{W}      &  \widetilde{B} \\
    \widetilde{B}^\top      & \widetilde{B}^\top \widetilde{W}^{-1} \widetilde{B}  
    \end{bmatrix}.
\end{equation}
To bound the error between exact and estimated matrices, we start by explicitly writing out $|| K - \widetilde{K} ||_2$ and using the triangle inequality to get the following,
\begin{align}
   || K - \widetilde{K} ||_2 &= 
    \left\|\begin{bmatrix}
    W - \widetilde{W}      &  B - \widetilde{B} \\
    B^\top - \widetilde{B}^\top      & C - \widetilde{B}^\top \widetilde{W}^{-1} \widetilde{B}  
    \end{bmatrix}\right\|_2 \label{eq:k_k_bound_expansion}\\
    &\leq \left\|\begin{bmatrix}
    W - \widetilde{W}        &  0 \\
    0    & 0
    \end{bmatrix}\right\|_2 +
    \left\|\begin{bmatrix}
    0    &  B - \widetilde{B}  \\
    B^\top - \widetilde{B}^\top       & 0 
    \end{bmatrix}\right\|_2 + 
    \left\|\begin{bmatrix}
    0      &  0 \\
    0  & C - \widetilde{B}^\top \widetilde{W}^{-1} \widetilde{B}  
    \end{bmatrix}\right\|_2 
\end{align}
\noindent where $C$ is the exact $(N-L)\times (N-L)$ kernel matrix of the set of data points not selected as landmark points, i.e. $C_{ij} = k(w_i , w_j )$ for $w_i \in \mathcal{L}'$. The first and second term can be reduced to $\norm{W - \widetilde{W}}_2$ and $\norm{B - \widetilde{B}}_2$ respectively. The last term can also be reduced to $\norm{C - \widetilde{B}^\top \widetilde{W}^{-1} \widetilde{B}}_2$ and bounded further, 
\begin{align}
    \norm{C - \widetilde{B}^\top \widetilde{W}^{-1} \widetilde{B}}_2 &=  \norm{(C - B^\top W^{-1} B) + (B^\top W^{-1} B - \widetilde{B}^\top \widetilde{W}^{-1} \widetilde{B})}_2 \\
    &\leq \norm{C - B^\top W^{-1} B}_2 + \norm{B^\top W^{-1} B - \widetilde{B}^\top \widetilde{W}^{-1} \widetilde{B}}_2\label{eq:c_bwb_bound}
\end{align}
\noindent The first term is precisely the classical bound on the Nystr\"{o}m approximation provided in Eq. \eqref{eq:bound_classical_nystrom}. We now obtain probabilistic bounds on the three terms $\norm{W - \widetilde{W}}_2$, $\norm{B - \widetilde{B}}_2$ and $\norm{B^\top W^{-1} B - \widetilde{B}^\top \widetilde{W}^{-1} \widetilde{B}}_2$ as a result of the quantum kernel estimations.
In this work, the estimation error occurs as a result of finite sampling error for each element in the kernel matrix. We can subsequently derive a bound on the normed difference between expected and estimated kernels by bounding the Bernoulli distribution of individual matrix element estimations. However, first we construct a bound on the spectral norm of the error in the kernel matrix through the following Lemma.

\begin{lemma}\label{lemma:spec_norm_bound}
Given a matrix $A\in\mathbb{R}^{d_1\times d_2}$ and the operator 2-norm (spectral norm) $||\cdot||_2$, we have, 
\begin{equation}\label{eq:spec_norm_bound}
    \Pr \big\{\norm{A}_2 \geq \delta\big\} \leq \sum_{i=1}^{d_1} \sum_{j=1}^{d_2} \Pr \Bigg\{|A_{ij}| \geq \frac{\delta}{\sqrt{d_1 d_2}} \Bigg\}
\end{equation}
where the matrix elements $A_{ij}$ of $A$ are independently distributed random variables.
\end{lemma}

\begin{proof}
We start by using the inequality $\norm{A}_2 \leq \norm{A}_\mathrm{F}$ where $\norm{\cdot}_\mathrm{F}$ refers to the Frobenius norm, 
\begin{align}
    \Pr \big\{\norm{A}_2 \geq \delta\big\} &\leq  \Pr \big\{\norm{A}_\mathrm{F} \geq \delta\big\} \\
    &= \Pr \Bigg\{\sum_{i=1}^{d_1} \sum_{j=1}^{d_2} |A_{ij}|^2 \geq \delta^2\Bigg\} \\
    &\leq \Pr \Bigg\{\bigcup_{i=1}^{d_1}\bigcup_{j=1}^{d_2} \Big( |A_{ij}|^2 \geq \frac{\delta^2}{d_1 d_2} \Big)\Bigg\} \\
    &\leq \sum_{i=1}^{d_1} \sum_{j=1}^{d_2} \Pr \Bigg\{|A_{ij}| \geq \frac{\delta}{\sqrt{d_1 d_2}} \Bigg\}
\end{align}
where the second line uses the definition of the Frobenius norm,  the third line uses the fact that of the $d_1 d_2$ summed terms, at least one must be greater than $\delta^2 / d_1 d_2$. The last inequality utilises the union bound.
\end{proof}

\noindent To bound the RHS of Eq. \ref{eq:spec_norm_bound} we use well known Hoeffding's ineqaulity.
\begin{theorem}\label{theorem:hoeffding}
(Hoeffding's inequality) Let $X_1, ..., X_M$ be real bounded independent random variables such that $a_i \leq X_i \leq b_i$ for all $i=1,...,M$. Then for any $t>0$, we have, 
\begin{equation}
    \Pr \Big\{ \Big\lvert \sum_{i=1}^M X_i  - \mu\Big\rvert \geq t\Big\} \leq 2 \exp \Bigg( -\frac{2t^2}{\sum_{i=1}^M (b_i - a_i)^2} \Bigg)
\end{equation}
where $\mu=\mathbb{E} \sum_{i=1}^M X_i$. 
\end{theorem}
\noindent Utilising Lemma \ref{lemma:spec_norm_bound} and Theorem \ref{theorem:hoeffding}, with probability at least $1-\Lambda$ we have the following bounds,
\begin{align}
    \norm{W - \widetilde{W}}_2 &\leq L \sqrt{\frac{1}{2M} \log\Big(\frac{2L^2}{\Lambda}\Big)}  \label{eq:w_bound}\\
    \norm{B - \widetilde{B}}_2 &\leq \sqrt{\frac{NL}{2M} \log\Big(\frac{2NL}{\Lambda}\Big)}  \label{eq:b_bound}
\end{align}

\noindent The last term $\norm{B^\top W^{-1} B - \widetilde{B}^\top \widetilde{W}^{-1} \widetilde{B}}_2$ is a bit more involved as we have matrix multiplication of random matrices. Furthermore, the existence of inverse matrices renders Hoeffding's inequality invalid as the elements of an inverse random matrix are not, in general, independent. Hence, we further this term to give,
\begin{align}
    \norm{B^\top W^{-1} B - \widetilde{B}^\top \widetilde{W}^{-1} \widetilde{B}}_2 &= \norm{B^\top \left(W^{-1} - \widetilde{W}^{-1}\right) B +B^\top \widetilde{W}^{-1} B - \widetilde{B}^\top \widetilde{W}^{-1} \widetilde{B}}_2 \\
    &\leq \norm{B^\top \left(W^{-1} - \widetilde{W}^{-1}\right) B}_2 +\norm{B^\top \widetilde{W}^{-1} B - \widetilde{B}^\top \widetilde{W}^{-1} \widetilde{B}}_2 \\
    &\leq \norm{W^{-1} - \widetilde{W}^{-1}}_2\norm{B}_2^2 +\norm{B^\top \widetilde{W}^{-1} B - \widetilde{B}^\top \widetilde{W}^{-1} \widetilde{B}}_2 \label{eq:getting_rid_of_inverse_prob}
\end{align}
\noindent where the last inequality uses the fact that for symmetric $Y$, we have the inequality $\norm{Z^\top YZ}_2 \leq \norm{Y}_2 \norm{Z}_2^2$, which can be shown by diagonalising $Y$ and using the fact that $\norm{ZZ^\top}_2 = \norm{Z}_2^2$. We first look at bounding the second term, which will also help illustrate the process by which equations \eqref{eq:w_bound} and \eqref{eq:b_bound} were obtained. We can expand the matrix multiplication such that we aim to bound the following, 
\begin{equation}\label{eq:prob_bound_big_term}
    \Pr \Big\{ \Big\lvert \Big(B^\top \widetilde{W}^{-1} B - \widetilde{B}^\top \widetilde{W}^{-1} \widetilde{B} \Big)_{ij}\Big\rvert \geq \delta'\Big\} = \Pr \Bigg\{ \Bigg\lvert \mu_{ij}^{\widetilde{W}} - \sum_{l, k =1}^L \widetilde{B}_{ki} \widetilde{W}^{-1}_{kl}\widetilde{B}_{lj} \Bigg\rvert  \geq \delta'\Bigg\}
\end{equation}
\noindent where for a specific $\widetilde{W}$ we have $\mu_{ij}^{\widetilde{W}}= \mathbb{E} \sum_{l, k =1}^L \widetilde{B}_{ki} \widetilde{W}^{-1}_{kl}\widetilde{B}_{lj} = (B^\top \widetilde{W}^{-1} B)_{ij}$. Since $\widetilde{B}_{ij}$ is estimated from $M$ Bernoulli trials we can write, $\widetilde{B}_{ij} = (1/M) \sum_{m=1}^M \widetilde{B}_{ij}^{(m)}$ where $\widetilde{B}_{ij}^{(m)}\in\{0,1\}$ . It is crucial to note that the inverse matrix $\widetilde{W}^{-1}_{ij}$ is fixed over the expression above -- as intended from Eq. \eqref{eq:getting_rid_of_inverse_prob}. We therefore take the elements $\{\widetilde{W}^{-1}_{ij}\}_{i,j=1}^{L}$ to be constant variables bounded by $0\leq \widetilde{W}^{-1}_{ij}\leq r_{\widetilde{W}}$ for all $i,j=1, ..., L$ and now expand expression \eqref{eq:prob_bound_big_term}:
\begin{align}
    \Pr \Bigg\{ \Bigg\lvert \mu_{ij} - \sum_{l, k =1}^L \widetilde{B}_{ki} \widetilde{W}^{-1}_{kl}\widetilde{B}_{lj} \Bigg\rvert  \geq \delta'\Bigg\} &= \Pr \Bigg\{ \Bigg\lvert \mu_{ij} - \sum_{l, k =1}^L \frac{1}{M^2}\Bigg(\sum_{m_1 =1}^M \widetilde{B}_{ki}^{(m_1)}\Bigg) \widetilde{W}^{-1}_{kl}\Bigg(\sum_{m_2 =1}^M \widetilde{B}_{lj}^{(m_2)}\Bigg) \Bigg\rvert  \geq \delta'\Bigg\} \\
    &:= \Pr \Bigg\{ \Bigg\lvert \mu_{ij} - \frac{1}{M^2}\sum_{l, k =1}^L \sum_{m=1}^{M^2} \omega^{(m)}_{iklj} \Bigg\rvert  \geq \delta'\Bigg\},
\end{align}
\noindent where in the second line we have expanded the sums over $m_1$ and $m_2$ into a single sum over $M^2$ terms. The random variables $\omega^{(m)}_{iklj}$ are bounded by $r$ and we utilise the concentration inequality in Theorem \ref{theorem:hoeffding} to bound the term above,
\begin{align}
    \Pr \Bigg\{ \Bigg\lvert \mu_{ij} - \sum_{l, k =1}^L \sum_{m=1}^{M^2} \frac{\omega^{(m)}_{iklj}}{M^2} \Bigg\rvert  \geq \delta'\Bigg\} &\leq 2 \exp \Bigg(-\frac{2 \delta'^2}{\sum_{l, k =1}^L \sum_{m=1}^{M^2} (r_{\widetilde{W}}/M^2)^2} \Bigg) \\
    &= 2 \exp \Bigg(-\frac{2 M^2 \delta'^2}{r_{\widetilde{W}}^2 L^2} \Bigg) .\label{eq:bound_with_hoeff}
\end{align}
\noindent We therefore have a bound on $\norm{B^\top \widetilde{W}^{-1} B - \widetilde{B}^\top \widetilde{W}^{-1} \widetilde{B}}_2$ by utilising Lemma \ref{lemma:spec_norm_bound}, and equations \eqref{eq:prob_bound_big_term},  \eqref{eq:bound_with_hoeff}, to give, 
\begin{equation}
    \Pr \Big( \norm{B^\top \widetilde{W}^{-1} B - \widetilde{B}^\top \widetilde{W}^{-1} \widetilde{B}}_2 \geq \delta \Big) \leq 2(N-L)^2 \exp\Bigg( -\frac{2M^2 \delta^2}{r_{\widetilde{W}}^2 L^2 (N-L)^2}\Bigg).
\end{equation}
\noindent Letting the RHS be equal to $\Lambda$ and solving for $\delta$, one can show with probability $1-\Lambda$,
\begin{equation}\label{eq:bwb_bwb_bound}
    \norm{B^\top \widetilde{W}^{-1} B - \widetilde{B}^\top \widetilde{W}^{-1} \widetilde{B}}_2 \leq \frac{r_{\widetilde{W}}L(N-L)}{M}\sqrt{\frac{1}{2} \log \Bigg( \frac{2(N-L)^2}{\Lambda}\Bigg)} .
\end{equation}
This bounds the second term in Eq. \eqref{eq:getting_rid_of_inverse_prob} leaving us with $\norm{W^{-1} - \widetilde{W}^{-1}}_2$, which can be bound,
\begin{align}
    \norm{W^{-1} - \widetilde{W}^{-1}}_2 &\leq c_{W^{-1}} c_{\widetilde{W}^{-1}}\norm{W-\widetilde{W}}_2 \\ 
    &\leq c_{W^{-1}} c_{\widetilde{W}^{-1}}\sqrt{\frac{L^2}{2M} \log\left(\frac{2L^2}{\Lambda}\right)} \label{eq:w_inv_minus_w_inv}
\end{align}
\noindent where we define $c_A:=\norm{A}_2$ for a given matrix $A$ with $A^{-1}$ being the pseudo-inverse for when $A$ is singular. The first inequality uses the fact that $\norm{W^{-1} - \widetilde{W}^{-1}} = \norm{W^{-1} (\widetilde{W}-W) \widetilde{W}^{-1}} \leq \norm{W^{-1}}\norm{\widetilde{W}^{-1}}\norm{\widetilde{W}-W}$ and the second inequality uses the bound in Eq. \eqref{eq:w_bound}.
Now putting together the bounds in equations \eqref{eq:k_k_bound_expansion}, \eqref{eq:c_bwb_bound}, \eqref{eq:b_bound}, \eqref{eq:w_bound}, \eqref{eq:bwb_bwb_bound}, \eqref{eq:getting_rid_of_inverse_prob}, \eqref{eq:w_inv_minus_w_inv} and \eqref{eq:bound_classical_nystrom}, we have,
\begin{align}
    || K - \widetilde{K} ||_2 \leq L \sqrt{\frac{1}{2M} \log\Bigg(\frac{2L^2}{\Lambda}\Bigg)} + \sqrt{\frac{NL}{2M} \log\Bigg(\frac{2NL}{\Lambda}\Bigg)}  + &\frac{r_{\widetilde{W}}L(N-L)}{M}\sqrt{\frac{1}{2} \log \Bigg( \frac{2(N-L)^2}{\Lambda}\Bigg)} \\
    &+c_B^2 c_{W^{-1}} c_{\widetilde{W}^{-1}}\sqrt{\frac{L^2}{2M} \log\left(\frac{2L^2}{\Lambda}\right)}
    + \mathcal{O}\Bigg(\frac{N}{\sqrt{L}}\Bigg) \nonumber\\
    \implies || K - \widetilde{K} ||_2 \leq \widetilde{\mathcal{O}} \Bigg(\frac{NL}{M} + &\frac{N}{\sqrt{L}}\Bigg)\label{eq:k_k_tilde}
\end{align}
\noindent where in the second line we use the fact that $N>>L$ and the notation $\widetilde{\mathcal{O}}$ to hide $\log$ dependencies. Furthermore, Eq. \eqref{eq:k_k_tilde} assumes that the matrix norms $c_B , c_{W^{-1}},c_{\widetilde{W}^{-1}}$ do not depend on the system dimension -- i.e. $L$ and $N$. For cases of $c_{W^{-1}}$ and $c_{\widetilde{W}^{-1}}$, since we are really taking the pseudo-inverse by selecting a $k$ rank approximation of $W$ with the $k$ largest positive eigenvalues, it is possible in some sense \textit{retain control} over these terms. The norm of $B$ on the other hand is entirely dependent on the dataset. The worst case scenario suggests a bound using the Frobenius norm, $c_B\leq ||B||_F =\sqrt{\sum_{i}^N\sum_{j}^L|B_{ij}|^2} \approx \sqrt{L(N-L)\overline{k}^2}$, where $\overline{k}$ is the average kernel element. This is however a loose upper bound that provides no further insight and it is often the case $\overline{k}<<1/N$ for quantum kernels.

One of the main objectives of Eq. \eqref{eq:k_k_tilde} is to understand the sampling complexity required to bound the error between estimated and real kernels. We find that to ensure this error is bounded we require $M\sim NL$. However, this does not account for the potentially vanishing kernel elements that we discuss further in \ref{sec:limitations_q_kernels}.
Eq. \eqref{eq:k_k_tilde} also highlights the competing interests of increasing $L$ to obtain a better Nystr\"{o}m approximation while at the same time benefiting from decreasing $L$ to reduce the error introduced from finite sampling many elements.

\subsubsection{Proof of Lemma \ref{lemma:pred_error_fin_samp}: Prediction error of SVM-NQKE due to finite sampling}\label{section:pred_error_fin_samp}

In the previous section, we obtained a bound to $\norm{K-\widetilde{K}}$ that suggested $M\sim NL$ to bound this term with respect to the error of finite samples that is introduced. This however does not identify the samples required in order to bound the error on the output of the SVM. In the noiseless case with Nystr\"{o}m approximation, given that we have hypothesis of the form, $h(x) = \mathrm{sign}\left(\sum_{i=1}^N \alpha_i k(x, x_i)\right)$, where $\alpha$ is obtained at training, let $f(x)=\sum_{i=1}^N \alpha_i k(x, x_i)$ so that $h(x) = \mathrm{sign}\left(f(x)\right)$. We equivalently let the SVM constructed through the estimated kernel be denoted as $\widetilde{h}(x) = \mathrm{sign}\left(\widetilde{f}(x)\right)$ where $\widetilde{f}(x)=\sum_{i=1}^N \alpha_i' k'(x, x_i)$. In this section, closely following \cite{10.1038/s41567-021-01287-z} we obtain a bound on $|f(x)-\widetilde{f}(x)|$ that will elucidate the samples required to suppress the shot noise of estimation. We start by expanding,
\begin{align}
    |f(x)-\widetilde{f}(x)| &= \left\lvert \sum_{i=1}^N \alpha_i' k'(x, x_i) - \sum_{i=1}^N \alpha_i k(x, x_i)  \right\rvert \\ 
    &\leq \sum_{i=1}^N \left\lvert \alpha_i' k'(x, x_i) -  \alpha_i k(x, x_i)  \right\rvert \\
    &= \sum_{i=1}^N \left\lvert (\alpha_i'-\alpha_i)  \left[k'(x, x_i)-k(x, x_i)\right] +  \alpha_i \left[k'(x, x_i)-k(x, x_i)\right]  + (\alpha_i'-\alpha_i) k(x, x_i)\right\rvert \\
    &\leq \norm{\alpha}_2 \cdot \norm{\nu(x)}_2 + \norm{\alpha'-\alpha}_2 \cdot \norm{\nu(x)}_2 + \sqrt{N}\norm{\alpha'-\alpha}_2 \label{eq:bound_on_f_f'},
\end{align}
\noindent where we let $\nu_i(x) =\nu_i= k'(x, x_i)-k(x, x_i)$, and where the last line uses the Cauchy-Schwarz inequality and the fact that $k(\cdot,\cdot)\leq 1 $. The term $\norm{\nu}_2$ is bounded in a similar fashion to \eqref{eq:getting_rid_of_inverse_prob}, as the kernel function is passed through the Nystr\"{o}m feature map (NFM). This gives $\norm{\nu}_2 \leq \mathcal{O}(\sqrt{L^2/2M}+ NL/M)$ where we assume $c_{k} , c_{W^{-1}},c_{\widetilde{W}^{-1}}$ to be bounded independent of $N$. Both terms, $\norm{\alpha}_2$ and $\norm{\alpha'-\alpha}_2 $, need slightly more work. 

The SVM dual problem in \eqref{eq:dual_lagrange} can be written in the following quadratic form for $\mu =0$ with regularisation $\lambda$, 
\begin{align}
    \text{minimise \ \ }& \frac{1}{2}\alpha^\top \left(Q + \frac{1}{\lambda}\mathbb{I}\right) - 1^\top \alpha \nonumber \\
    \text{subject to: \ \ }&\alpha_i \geq 0, \forall i=1, ...., N. \label{eq:dual_quad_program}
\end{align}
where $Q_{ij}=y_i y_j K_{ij}$. This allows us to use the following robustness Lemma on quadratic programs.
\begin{lemma}\label{lemma:perturb_quad_program}
(\cite{10.1007/BF01580110}, Theorem 2.1) Given a quadratic program of the form in \eqref{eq:dual_quad_program}, let $Q'$ be a perturbation of $Q$ with a solution of $\alpha'$ and $\alpha$ respectively. Given that $\norm{Q'-Q}_\mathrm{F}\leq \epsilon < \lambda_{\mathrm{min}}$, where $\lambda_{\mathrm{min}}$ is the minimum eigenvalue of $Q+\frac{1}{\lambda}\mathbb{I}$, then, 
\begin{equation}
    \norm{\alpha'-\alpha}_2 \leq \epsilon (\lambda_{\mathrm{min}}- \epsilon)^{-1}\norm{\alpha}_2.
\end{equation}
\end{lemma}
\noindent Note, this Lemma assumes small perturbations and breaks down for when $\epsilon\geq\lambda_{\mathrm{min}}$. It can be shown from analysing the solutions from the KKT conditions that $\mathbb{E}[\norm{\alpha}_2^2] = \mathcal{O}(N^{2/3})$ -- see \cite{10.1038/s41567-021-01287-z} (Lemma 16) for further details. This finally leaves us with bounding $\norm{Q' - Q}_\mathrm{F}$. We know, 
\begin{align}
    \norm{Q' - Q}_\mathrm{F}^2 &= \sum_{ij} |Q'_{ij} - Q_{ij}|^2 = \sum_{ij} |y_i y_j K'_{ij} - y_i y_j K_{ij}|^2 \\ 
    &= \sum_{ij} |K'_{ij} - K_{ij}|^2 = \norm{K' - K}_\mathrm{F}^2.
\end{align}
\noindent Since we are bounding the error between estimated and real matrices (both including Nystr\"{o}m approximation), we in fact require the bound $\norm{\widetilde{K} - \widehat{K}}_\mathrm{F}$ from \eqref{eq:nystrom_kernel} and \eqref{eq:quantum_nystrom_kernel}, i.e. we have $K=\widehat{K}$ and $K'=\widetilde{K}$. We compute this in Section \ref{sec:proof_k_k_tilde} to be $\norm{K' - K}_\mathrm{F}\leq \mathcal{O}\left(\frac{NL}{M} +  \sqrt{\frac{NL}{M}}\right)$ -- noting that we always weakly bound the spectral norm with the Frobenius norm, $\norm{\cdot}_2\leq \norm{\cdot}_\mathrm{F}$. Therefore, using Lemma \ref{lemma:perturb_quad_program} with the fact that $\lambda_\text{min}\geq 1/\lambda$ for constant $\lambda$, we have,
\begin{equation}
    \norm{\alpha'-\alpha}_2 \leq \mathcal{O}\left(\frac{N^{5/6}\sqrt{L}}{\sqrt{M}}\right).
\end{equation}
\noindent Substituting these results into Eq. \eqref{eq:bound_on_f_f'}, we obtain the bound on the error of the model output,
\begin{equation}\label{eq:end_proof_lemma_f_f_tilde}
    |f(x)-\widetilde{f}(x)| \leq \mathcal{O}\left( \frac{N^{4/3}\sqrt{L}}{\sqrt{M}} \right) = \mathcal{O}\left( \sqrt{\frac{N^{8/3}L}{M}} \right),
\end{equation}
\noindent as required. Therefore we see that $M\sim N^3 L$ is sufficient to suppress error in the model output. Note, Eq. \eqref{eq:end_proof_lemma_f_f_tilde} only shows the term most difficult to suppress with respect to $M$.

\subsection{Processing datasets for numeric analysis}\label{section:data_and_relabel_process}

This work observes the performance of the QRF on five publicly available datasets. In certain cases, continuous features were selected as class labels by introducing a separating point in $\mathbb{R}$ such as to split the number of data points evenly in two. A summary of the original datasets are presented in the table below.

\setlength\tabcolsep{5.5pt}
\begin{center}
\begin{tabular}{ c  c c  c   p{5.8cm}  p{2.8cm} }
 \multicolumn{6}{c}{Datasets} \\
 \hline
 Name & Denoted & \# of points & $|\mathcal{X}|$ & Class labels & Reference \\
 \hline
 \textit{Fashion MNIST} & $\mathcal{D}_{\mathrm{FM}}$ & 70,000 & $28\times 28$ & Originally $10$ classes with $7000$ per class. Binary class transformation uses class $0$ (top) and $3$ (dress). & Zolando's article images \cite{https://doi.org/10.48550/arxiv.1708.07747} \\
 \textit{Breast Cancer} &  $\mathcal{D}_{\mathrm{BC}}$ & 569 & 30 & This is a binary class dataset. & UCI ML Repository \cite{Dua2019}\\
 \textit{Heart Disease}&  $\mathcal{D}_{\mathrm{H}}$ & 303 & 13 & The class is defined as $0$ and $1$ for when there is respectively a less than or greater than $50\%$ chance of disease. & UCI ML Repository \cite{Dua2019}\\
\textit{Ionosphere} & $\mathcal{D}_{\mathrm{Io}}$  & 351 & 34 & This is a binary class dataset. &UCI ML Repository \cite{Dua2019}\\
 \hline
\end{tabular}  
\end{center}

\noindent Note: $|\mathcal{X}|$ indicates the dimension of the feature space after accounting for any removal of categorical features and the possible use of a feature to create class labels. Finally, the analysis of various QRF parameters use the Fashion MNIST dataset. To ensure the analysis has tractable computation, $300$ points are sampled for each training and testing (unless otherwise stated) -- which partially accounts for the variance in accuracies observed.

\subsubsection{Data pre-processing}\label{section:data_preproc}

Given the original dataset $\mathcal{S}_{\mathrm{tot}}'' = \{ (x_i'', y_i)\}_{i=1}^N$ where $x_i'' \in \mathbb{R}^{D'}$ and $y_i \in \{-1, +1\}$ in this section we show the pre-processing steps undertaken prior to the QRF learning process. The first step mean-centres each of the features of the dataset and employs \textit{principal component analysis} \cite{doi:10.1080/14786440109462720} to obtain $D<D'$ dimensional points, giving a the dataset $\mathcal{S}_{\mathrm{tot}}'= \{ (x_i', y_i)\}_{i=1}^N$ where $x' \in \mathbb{R}^D$. In practice, this dimensionality reduction technique is a common pre-processing step for both classical and quantum ML. The technique selects the $D$-most primary directions in which the data is spread and subsequently projects the data points onto the plane spanned by those directions -- see \cite{10.1007/s00362-019-01124-9} for more details. This gives us control over the dimensionality of the data given to the QRF -- this being especially important with the $\Phi_{\mathrm{IQP}}$ embedding over $n$ qubits requiring an exact dimension of $n$. Furthermore, PCA relies on the principle that data often lies in a lower dimensional manifold \cite{10.5555/3086952} hence, resulting in PCA reduction often aiding the learning model to avoid noise. Therefore it becomes crucial to include the pre-processing step of a PCA reduction when comparing between learning models. This highlights observation in the main text where we see a degradation in performance when increasing number of qubits for the Fashion MNIST data. Though it may seem as though we can make a statement about the ineffectiveness of quantum kernels with increasing qubit numbers, this misses the effect of PCA in the pre-processing step. This is why both the classical and quantum models struggle when $D$ is increased.  

The second step in pre-processing is unique to quantum models. Before data can be inserted into the QRF algorithm, we are required to process these classical vectors $x' \in \mathbb{R}^D$ so that features can be effectively mapped to single qubit rotations. This is done by observing each dimension (feature), $\Omega_i$, of the space consumed by the set of all data points (including training and testing) and normalising so that the range of values lie between $0$ and $\pi$. Mathematically we make the transformation, $x_i' \xrightarrow{} x_i = \pi (r^{\mathrm{min}}_i + x') /(r^{\mathrm{max}}_i - r^{\mathrm{min}}_i)$ where $r^{\mathrm{min}}_i = \min \{x' | (x' , y) \in \mathcal{S}_{\mathrm{tot}}' \}$ and $r^{\mathrm{max}}_i$ is defined similarly. The final dataset is therefore given by $\mathcal{S}_{\mathrm{tot}} = \{ (x_i, y)\}_{i=1}^N$ where $x_i \in [0, \pi]^{D}$ and $y_i \in \{-1, +1\}$. The embedding into half a rotation $[0, \pi]$, as opposed to full rotations of $ [0, 2\pi]$ or $ [-\pi, \pi]$, is to ensure that the ends of the feature spectrum, $r^{\mathrm{min}}$ and $r^{\mathrm{max}}$, are not mapped to similar quantum states.

\subsubsection{Relabelling datasets for optimal quantum kernel performance, $\mathscr{R}_{\Phi}^{\text{QK}}$}\label{section:relab_qk}

To demonstrate that QML is able to achieve an advantage over it's classical counterpart, a crucial first-step would be to show that there at least exist \textit{artificial} datasets in which quantum models outperform classical methods. We generate such a dataset by relabelling a set of real data points so that a classical kernel struggles and quantum kernel does not. This is done by maximising the model complexity -- and hence generalisation error -- of the classical kernel while at the same time minimising the model complexity of the quantum kernel. This is precisely the strategy followed in \cite{10.1038/s41467-021-22539-9}, where the regime of potential quantum advantage was posited to be related to -- what the authors referred to as -- the \textit{geometric difference}, $G$, between quantum and classical models. This quantity was shown to provide an upper-bound to the ratio of model complexities -- more specifically, $G_{CQ} \geq \sqrt{s_C / s_Q}$. The task of relabelling is therefore to solve the following optimisation problem,
\begin{equation}\label{eq:relabel_maximisation}
    Y^* = \argmax_{Y\in\{0,1\}^N} \frac{s_C}{s_Q} = \argmax_{Y\in\{0,1\}^N} \frac{Y^\top (K^C)^{-1}Y}{Y^\top (K^Q)^{-1}Y}
\end{equation}
\noindent where $Y^* =[y_1 , ..., y_N]$ are the optimal relabelled instances and $K^C$, $K^Q$ are Gram matrices for the classical and quantum kernels, respectively. This is related to the process of maximising the \textit{kernel target alignment} measure \cite{Cristianini2006}, where in order for a kernel to have high performance, the targets, $Y$, must be well \textit{aligned} with the transformed set of instances. The quantity is more intuitive and has the form, 
\begin{equation}\label{eq:target_alignment}
    \mathcal{T}(K) = \frac{\sum_{i=1}^N\sum_{j=1}^NK_{ij}y_i y_j }{N\sqrt{\sum_{i=1}^N\sum_{j=1}^N K_{ij}^2}} = \frac{Y^\top K Y}{N\norm{K}_F}
\end{equation}
\noindent where $||\cdot||_F$ is the Frobenius norm. Therefore, the equivalent optimisation problem of Eq. \eqref{eq:relabel_maximisation} becomes the minimisation of the quantity $\mathcal{T}(K^C)/\mathcal{T}(K^Q)$. Nevertheless, the optimisation problem in Eq. \eqref{eq:relabel_maximisation} can be transformed to a \textit{generalised eigenvalue problem} of the following form by replacing $Y\xrightarrow{}\phi\in\mathbb{R}^N$,
\begin{equation}
    (K^C)^{-1} \phi = \lambda (K^Q)^{-1} \phi
\end{equation}
\noindent where $\phi$ and $\lambda$ are the eigenvectors and eigenvalues respectively. As \eqref{eq:relabel_maximisation} is a maximisation problem the solution is therefore the eigenvector $\phi^*$ associated with the largest eigenvalue $\lambda^*$. To obtain an exact solution, we use the fact that the generalised eigenvalue decomposition is related to the regular decomposition of $\sqrt{K^Q}(K^C)^{-1}\sqrt{K^Q}$ with identical eigenvalues $\lambda_i$. Hence, supposing a decomposition of $\sqrt{K^Q}(K^C)^{-1}\sqrt{K^Q}=Q\Lambda Q^\top$, where $\Lambda = \mathrm{diag}(\lambda_1, ..., \lambda_N)$, it can be shown that $\phi^* = \sqrt{K^Q}q^*$ \cite{boyd_vandenberghe_2004}, where $q^*$ is the eigenvector associated with the largest eigenvalue, $\lambda^*$. The continuous vector $\phi^*$ can then be transformed into labels by simply taking the sign of each entry. Furthermore, following \cite{10.1038/s41467-021-22539-9} we also inject randomness, selecting the label -- for a given data point $x_i$ -- as $y_i = \mathrm{sign}(\phi_i)$ with probability $0.9$ and randomly assigning $y_i = \pm 1$ otherwise. This is relabelling strategy is denoted as $(y_i)_i = \mathscr{R}_{\Phi}^{\text{QK}}(\{(x_i, y_i')\}_i)$.

It is important to note, the relabelling of datasets is dependent on the specific quantum and classical kernels employed. In this work, the \textit{radial basis function } kernel is used as the classical kernel and we obtain different labels depending on the quantum embedding used, $\Phi_{\mathrm{IQP}}$ or $\Phi_{\mathrm{Eff}}$. 

\subsubsection{Relabelling datasets for optimal QRF performance, $\mathscr{R}_{\Phi}^{\text{QRF}}$}\label{section:relab_qrf}

Classically, SVMs require kernels in cases where points in the original feature space are not linearly separable. A simple example is shown in Figure \ref{fig:qrf_relab_and_qrf_advan}a where any single hyperplane in this space cannot separate the two classes. However, it is clear that a tree structure with multiple separating hyperplanes has the ability to separate the clusters of points. This inspires a relabelling strategy that attempts to separate instances into four regions using a projection to a particular axis. Since we are able to only compute the inner products of instances and do not have access to the direct induced Hilbert space, we randomly select two points, $x'$ and $x''$, to form the axis of projection. Projection of point $z$ onto this axis is given by the elementary vector projection, 
\begin{equation}\label{eq:vector}
    \frac{\langle z - x', x'' - x'\rangle_{\mathcal{H}}}{\sqrt{\lvert\langle x''-x', x''-x'\rangle_{\mathcal{H}}\rvert}}
\end{equation}
\noindent where we have specified the inner product over the reproducing Hilbert space $\mathcal{H}$ associate with some quantum kernel $k$. Using the fact that points $x'$ and $x''$ are fixed, we define a projection measure of the form,
\begin{equation}
    P(x_i) =  \langle x_i , x'' \rangle_{\mathcal{H}} - \langle x_i , x' \rangle_{\mathcal{H}} = k(x_i , x'') - k(x_i , x') := P_i
\end{equation}
\noindent where $P_i := P(x_i)$ for $x_i \in \mathcal{S}|_{\mathcal{X}}$. We now label the classes such that we have alternating regions,
\begin{equation}
    Y_i = 
    \begin{cases}
      0 & \text{if } P_i< P^{(\mathrm{q1})} \text{ or }  P^{(\mathrm{q2})} \leq P_i <P^{(\mathrm{q3})}\\
      1 & \text{otherwise} 
    \end{cases}
\end{equation}
\noindent where $P^{(\mathrm{q1})}, P^{(\mathrm{q2})}, P^{(\mathrm{q3})}$ are the first, second and third quartiles of $\{P_i\}_{i=1}^N$, respectively. An illustration of this relabelling process can be seen in Figure \ref{fig:qrf_relab_and_qrf_advan}a. One should note, unlike the relabelling strategy of Section \ref{section:relab_qk}, here we do not compare against the performance of the classical learner. The intuition -- and what we observe numerically -- is that a clear division in the quantum induced Hilbert space does not translate to a well-defined division employing classical kernels. Therefore we have a labelling scheme that is in general both hard for classical learners as well as QSVMs. This relabelling process will be referred to as $(y_i)_i = \mathscr{R}_{\Phi}^{\text{QRF}}(\{(x_i, y_i')\}_i)$.

\begin{figure}
    \centering
    \includegraphics[width=440pt]{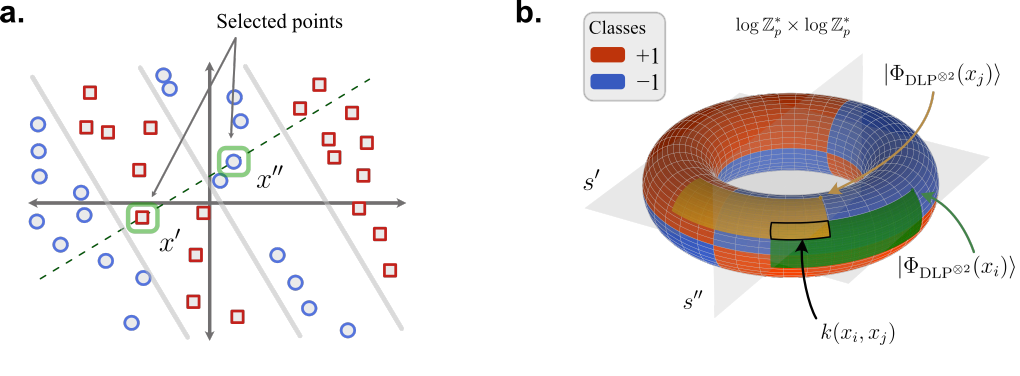}
    \vspace{-0.4cm}
    \caption{Both Figures depict concept classes that are not learnable by simple hyperplanes. \textbf{(a)} Two points are randomly selected (green) for the relabellingn process $\mathscr{R}_\Phi^{\text{QRF}}$ elaborated in \ref{section:relab_qrf}. This creates an alternating pattern in $\mathcal{H}_\Phi$. Note, in high dimension points become sparsely separated, resulting in only slight advantage of the QRF against other learners. \textbf{(b)} Here we observe an extension of the DLP classification problem \cite{10.1038/s41567-021-01287-z} into two dimensions. The torus is split into four alternating regions at $s'$ and $s''$. Since this separation exists in log-space, classical learners are presumed unable to solve this problem due to the hardness of solving DLP. Furthermore, since multiple hyperplanes are required, a simple quantum kernel also fails. }
    \label{fig:qrf_relab_and_qrf_advan}
\end{figure}

\subsection{Learning advantage with a QRF}\label{section:qrf_advantage}

In the main text we asserted that the QRF can extend the linear hypotheses available to linear quantum models. In this section we show that there are problems that are both difficult for quantum kernels (and by extension QNNs \cite{qnns_are_kernel}) and classical learners. We claim advantages using the QRF model by observing the hypothesis set available as a result of the QDT structure. We assume that we have $L=N$ models. In other words, we look at QDTs that do not take a low-rank Nystr\"{o}m approximation.  

To show there exist concept classes that are unlearnable by both classical learners and quantum linear models we extend the DLP-based concept class given in \cite{10.1038/s41567-021-01287-z}. The work looks at the class $\mathscr{C}_{\text{QK}}= \{f_s^{\text{QK}}\}_{s\in \mathbb{Z}_p^*}$ where, 
\begin{equation}
    f_s^{\text{QK}}(x) = \begin{cases}
      +1 & \text{if} \log_g x \in [s, s+\frac{p-3}{2}] \\
      -1 & \text{otherwise}, 
    \end{cases}
\end{equation}
\noindent and $x\in \mathbb{Z}_p^*$, for generator $g$ of $\mathbb{Z}_p^* = \{1,2,...,p-1\}$ with large prime $p$. Each concept of the concept class divides the set $\mathbb{Z}_p^*$ into two sets, giving rise to a binary class problem. It is important to note, the problem is a trivial $1$D classification task on the assumption one can efficiently solve the DLP problem. However, it can be shown that the DLP problem lies in BQP and is otherwise presumed to be classically intractable. This gives rise to a quantum advantage in learning the concept class $\mathscr{C}_{\text{QK}}$ with a quantum kernel with the following quantum feature map:
\begin{equation}
    \ket{\Phi_{\text{DLP}}^q(x_i)} = \frac{1}{\sqrt{2^q}} \sum_{j\in \{0,1\}^q} \ket{x_i \cdot g^j} 
\end{equation}
\noindent where $x_i \in \{0,1\}^n$ over $n=\lceil \log_2 p \rceil$ qubits, and $q=n-t\log n$ for some fixed $t$. This state can be efficiently prepared with a fault-tolerant quantum computer using Shor's algorithm \cite{Shor_1997}. Such embeddings can be interpreted as \textit{interval states}, as they are a superposition over an interval in log-space, $[\log x , \log x + 2^q - 1]$. The subsequent kernel function, $k(x_i, x_j) = |\langle \Phi_{\text{DLP}}^q(x_j) | \Phi_{\text{DLP}}^q(x_i)\rangle|^2$ can therefore be interpreted as the interval intersection of the two states in log-space. Guarantees of the inability to classically estimate these kernel elements up to additive error follow identically to proofs supplied in \cite{10.1038/s41567-021-01287-z}. 

Now, to show that there exist advantages of using a QRF, we construct a different concept class, $\mathscr{C}_{\text{QRF}}= \{f_{s',s''}^{\text{QRF}}\}_{s', s'' \in \mathbb{Z}_p^*}$ where we have, 
\begin{equation}
  f_{s',s''}^{\text{QRF}}(x) =   \begin{cases}
      +1 & \text{if} \log_g x^{(0)} \in [s', s'+\frac{p-3}{2}] \text{ and } \log_g x^{(1)} \notin [s'', s''+\frac{p-3}{2}]\\
      +1 & \text{if} \log_g x^{(0)} \notin [s', s'+\frac{p-3}{2}] \text{ and } \log_g x^{(1)} \in [s'', s''+\frac{p-3}{2}]\\
      -1 & \text{otherwise}
   \end{cases},
\end{equation}
\noindent and where $x=(x^{(0)}, x^{(1)})\in \mathbb{Z}_p^*\times \mathbb{Z}_p^*$ for generator $g$ of $\mathbb{Z}_p^*$. This can be interpreted as a separation of a torus into four segments as shown in Figure \ref{fig:qrf_relab_and_qrf_advan}b. The associated quantum feature map can be constructed through the tensor product, $\ket{\Phi_{\text{DLP}^{\otimes 2}}^q(x_i)} = \ket{\Phi_{\text{DLP}}^q(x_i^{(0)})}\ket{\Phi_{\text{DLP}}^q(x_i^{(1)})}$ over $n=2\lceil \log_2 p \rceil$ qubits. This gives a kernel function of the form, $k(x_i, x_j) = |\langle \Phi_{\text{DLP}}^q(x_j^{(0)}) | \Phi_{\text{DLP}}^q(x_i^{(0)})\rangle|^2\cdot|\langle \Phi_{\text{DLP}}^q(x_j^{(1)}) | \Phi_{\text{DLP}}^q(x_i^{(1)})\rangle|^2$ which by Theorem \ref{theorem:prod_kernels} is a valid \textit{product kernel} of DLP kernels. The interpretation of this kernel function is now a $2$D interval (area) in log-space. 

Since concepts $\mathscr{C}_{\text{QRF}}$ can not be learned with a single hyperplane in log-space, a quantum kernel with a $\Phi_{\text{DLP}^{\otimes 2}}^q$ embedding will be unable to learn this concept. Classical methods also fail due to the assumed hardness of the DLP problem. In comparison, a QDT with its tree structure, is able to construct multiple decision boundaries to precisely determine the regions of different classes. In practice, there is high likelihood of any single QDT overfitting to the given samples -- hence requiring an ensemble of QDTs to form a QRF model. 

Finally, it important to note that we are only claiming that the concept class  $\mathscr{C}_{\text{QRF}}$ is unlearnable by QSVMs and \textit{potentially} learnable by a QRF. There clearly exist a hypothesis $h_c \in \mathbb{H}_{\text{QDT}}$ that emulates a particular concept  $c\in \mathscr{C}_{\text{QRF}}$. However, this does not say anything about the ability of the QRF algorithm to obtain $h_c$ with a particular training set sampled from $c$. Hence there remains the question of whether there exists an evaluation algorithm $\mathcal{A}$, such that given a training set $\mathcal{S}_c = \{x_i, c(x_i)\}_{i=1}^N$, we have $\mathcal{A}(\mathcal{S}_c)|_{x=x'}=h(x')\approx c(x')$ for some predefined $c\in \mathscr{C}_{\text{QRF}}$.

\subsection{Limitations of quantum kernel methods}\label{sec:limitations_q_kernels}

Recent work \cite{https://doi.org/10.48550/arxiv.2207.05865} has shown that quantum kernels can be constructed to present quantum advantage using any BQP decision problem; however, there is a lack of evidence that it provides any  benefit when dealing with real-world data. Though they are crucial milestones in the development of QML, such constructions -- including the advantage shown in \cite{10.1038/s41567-021-01287-z} -- embed solutions to the learning problem with the specific embedding employed. 
It is hard to see the set of model hypotheses, arising from the DLP embedding \cite{10.1038/s41567-021-01287-z}, be useful in any other problem disregarding the specific example it was designed for. Therefore, as was required in the development of effective classical kernels, further research is required to construct relevant embeddings for practical problems \cite{inductive_bias_of_qk}.

Obtaining appropriate embeddings, one needs to also consider the \textit{curse of dimensionality} \cite{bellman1957dynamic}. With greater dimensions, a larger number of samples of data points are required to understand patterns in this space. Though obtaining a separation between sparsely distributed points is generally easier (Cover's theorem \cite{10.1109/pgec.1965.264137}), they are far more prone to overfitting. Therefore, generalisability seems to be unlikely with the exponential growth of the Hilbert space with respect to the number of qubits. Furthermore, as the dimensions increase, points are sent to the peripheries of the feature space, resulting in the inner product becoming meaningless in high dimension. This results in kernel matrix elements being exponentially (with respect to the number of qubits) suppressed to zero. Hence, each kernel element would require an exponentially large number of samples to well-approximate the vanishing kernel elements. Though it was seen that the Nystr\"{o}m approximation allowed for a smaller complexity of shots to bound the error in the kernel matrix, $M\sim NL$ (see Section \ref{sec:error_bounds}), this does not take into account a vanishing kernel element. Bernoulli trials have variance, $\frac{(1-k)k}{M}$ for kernel element $k$. This means that for $k\sim 2^{-n}$ we are required to have $M\sim \mathcal{O}(2^n)$ to distinguish elements from zero. This was discussed in \cite{10.1038/s41534-021-00498-9} and specific embeddings were constructed so that kernel elements did not vanish with larger numbers of qubits. To further aide with this problem, reducing the bandwidth parameter of quantum kernels is seen to improve generalisation \cite{kernel_band_shaydulin, band_enabl_gen}. In essence, both these methods work to reduce the volume of Hilbert space explored, resulting in larger inner products.

Finally, there is also the practical questions of the effect of noise and finite sampling error on the kernel machine optimisation. The SVM problem is convex, however this is not true once errors are incorporated. In \cite{10.1038/s41567-021-01287-z}, the effect of such errors to the optimal solution are explored. It is shown that the noisy halfspace learning problem is robust to noise with $M\sim \mathcal{O}(N^4)$. However, this is an extremely large number of shots that very quickly becomes impractical for even medium-size datasets. It is crucial that work is done to reduce this complexity if we are to see useful implementations of kernel methods. In the case of the QRF, we welcome weak classifiers and there are arguments for such errors being a natural \textit{regulariser} to reduce overfitting \cite{noisy_quantum_kernel_machines}.

\newpage

\section{Summary of Variables}
There are many moving parts and hyperparameters to the QRF model. We therefore take the chance to list out most of the parameters used throughout the paper. Hopefully this saves many the headache of scavenging for definitions.

\vspace{0.3cm}

\begin{center}
\begin{tabular}{ |c   p{5cm}||c p{5cm}|  }
 \hline
 \multicolumn{4}{|c|}{Parameter List} \\
 \hline 
 \hline
 $\mathscr{C}$ & Concept class. See Section \ref{section:qrf_advantage} for concept class based on DLP.&
 $C$ & Penalty parameter for fitting SVM, $C\propto 1/\lambda$. \\
 $\mathcal{C}$ &  The set of classes. & 
  $d$ & Maximum depth of QDT.  \\
  $D$ & Dimension of the original classical feature space, $x_i \in \mathbb{R}^D$ or equivalently $D=\mathrm{dim}(\mathcal{X})$. & 
  $\mathcal{D}$ & Indicates a specific dataset with $\mathcal{S}=\{(x_i , y_i)|(x_i , y_i)\sim \mathcal{D}\}_{i=1}^N $ forming the training set.  \\
 $G_{KW}$ & \textit{Geometric difference} between kernels $K$ and $W$ -- see discussion in Section \ref{section:relab_qk}. & 
 $\mathcal{H}$ & Hilbert space. $\mathcal{H}_\Phi$ refers to the quantum feature space induced by the embedding $\Phi$.\\
 $\mathbb{H}$ & Hypothesis set - see Section \ref{section:gen_err_bound_clas}. & 
 $k(x',x'')$ & Kernel function. \\
 $K$ & Gram matrix or also referred to as the kernel matrix, $K_{ij}=k(x_i, x_j)$. & 
 $L$ & Number of landmark points chosen for Nystr\"{o}m approximation -- see Section \ref{section:nystrom_intro}. \\
 $M$ & Number of samples taken from each quantum circuit (shots). & 
 $n$ & Number of qubits. \\
 $N$   & Total number of training instances.    &
 $\mathcal{N}_{\Phi, L}^{(i)}$   & Denotes $i^\text{th}$ node with embedding $\Phi$ and $L$ landmark points.     \\
 $\mathbf{N}(\cdot)$  & Nystr\"{o}m feature map. See definition in Eq. \eqref{eq:nystrom_feature_map}.    & 
 $\mathcal{R}(h)$  & Generalisation error for hypothesis $h$ -- see Def. \ref{def:gen_error}.  $\widehat{\mathcal{R}}(h)$ refers to the empirical error.   \\
 $\mathfrak{R}_{\mathcal{S}}$ & Rademacher complexity over dataset $\mathcal{S}$. The Empirical Rademacher complexity is then given by $\widetilde{\mathfrak{R}}_{\mathcal{S}}$, see Section \ref{section:gen_err_bound_clas}.    & $\mathscr{R}_\Phi$ & Relabelling function for a given embedding $\Phi$. Can either be $\mathscr{R}^{\text{QK}}_\Phi$ or $\mathscr{R}^{\text{QRF}}_\Phi$ with explanations in Sections \ref{section:relab_qk} and \ref{section:relab_qrf} respectively. \\
  $s_K$ & Model complexity for kernel $K$ defined as $s_K = Y^\top K^{-1}Y$ for given training labels $Y$.  &
  $\mathcal{S}^{(i)}$ &Training set available at node $i$. $\mathcal{S}_{\mathrm{tot}}$ refers to all data points (including training and testing) from a given dataset. $\mathcal{S}|_{\mathcal{X}}$ refers to the data vectors. \\
 $T$ &Number of trees in the QRF.  &
  $\mathcal{T}(K)$ & Kernel target alignment measure given Gram matrix $K$. Quantity defined in Eq. \eqref{eq:target_alignment}.\\
 $\lambda$ & Regularisation parameter for fitting SVM, $\lambda \propto 1/C$. & $\Phi$ & Indicates the quantum embedding and takes the form $\Phi_{\mathrm{IQP}}$ or $\Phi_{\mathrm{Eff}}$ defined in Section \ref{section:embedding}. \\
 $\gamma_i$ & Margin at node labelled $i$. &
$F^{(i)}(\cdot)$ & Pseudo class map that defines the binary class for the splitting at node $i$. Refer to Section \ref{sec:multiclass_classification}. \\
$\mathcal{X}$ & Feature space.   &
$\mathcal{Y}$ & The set of class labels.   \\
 \hline
\end{tabular}  
\end{center}

\end{document}